%% file: main.tex
\documentclass[a4paper,UKenglish,cleveref, autoref, thm-restate]{lipics-v2021}

\pdfoutput=1

\usepackage[utf8]{inputenc}

\title{Categorical Semantics for Model Comparison Games for Description Logics}

\author{Mateusz Urbańczyk}{ Institute of Computer Science, University of Wrocław, Poland }{ mateusz.urbanczyk97@gmail.com}{}{}

\authorrunning{M. Urbańczyk}

\Copyright{Mateusz Urbańczyk}



\category{Master Thesis}

\relatedversiondetails[linktext={FLoC, DL 2022}]{}{http://ceur-ws.org/Vol-3263/paper-5.pdf}

\acknowledgements{
    Formal advisor of this master thesis was Emanuel Kieroński, and informal Bartosz Bednarczyk.
    Results stated here were previously published on
    DL 2022, 35th International Workshop on Description Logics and presented
    during Federated Logic Conference (FLoC) 2022 at Haifa, Israel~\cite{dl-comonadic-games-paper},
    this is an extended and revised version. \\ \\
    I would like to thank Bartosz Bednarczyk for proposing me this topic once
    on a university corridor and for extensive support throughout the project.
    I would also like to thank Maciej Piróg and Emanuel Kieroński, for introducing
    me to the deep world of category theory and decidability,
    which allowed me to tackle problems approached while writing this thesis.
}

\nolinenumbers %uncomment to disable line numbering

\hideLIPIcs

\input{settings.tex}
\input{macros.tex}

\begin{document}

\maketitle

\begin{abstract}
    A categorical approach to study model comparison games in terms of comonads was
    recently initiated by Abramsky et al. In this work, we analyse games that
    appear naturally in the context of description logics and supplement them with suitable game comonads.
    More precisely, we consider expressive sublogics of $\ALCOIbSelf$, namely,
    the logics that extend $\ALC$ with any combination of inverses, nominals,
    safe boolean roles combinations, and $\Self$ operator. Our construction
    augments and modifies the so-called modal comonad by Abramsky and Shah. The
    approach that we took heavily relies on the use of relative comonads, which
    we leverage to encapsulate additional capabilities within the bisimulation
    games in a compositional manner.
\end{abstract}

\input{1-introduction}

\input{2-preliminaries}

\input{3-bisimulation-games}

\input{4-reductions}

\input{5-game-comonads}

\input{6-extensions}

\input{7-conclusions}

\bibliography{references}

\end{document}

%% file: settings.tex
\usepackage{listings}
\usepackage{ellipsis}         % improves whitespacing around "..."s
\usepackage[abbreviations,british]{foreign} % for better typing i.e., e.g. etc. 
\usepackage{etoolbox}
\usepackage{etex}

\definecolor{darkmidnightblue}{rgb}{0.0, 0.2, 0.4}
\definecolor{persianplum}{rgb}{0.44, 0.11, 0.11}

\hypersetup{
  colorlinks  = true, % Colours links instead of ugly boxes
  citecolor   = persianplum, % Colour of citations
  urlcolor    = persianplum, % Colour for external hyperlinks,
  linkcolor   = persianplum
}

\usepackage{tikz}
\usepackage{todonotes}
\usepackage{xspace}

\definecolor{ao(english)}{rgb}{0.0, 0.5, 0.0}
\definecolor{brickred}{rgb}{0.8, 0.25, 0.33}

\usepackage{amssymb}
\usepackage{amsmath}
\usepackage{amsthm}
\usepackage{mathtools}

\usepackage{physics}
\usepackage{amsmath}
\usepackage{tikz}
\usepackage{mathdots}
\usepackage{yhmath}
\usepackage{cancel}
\usepackage{color}
\usepackage{siunitx}
\usepackage{array}
\usepackage{multirow}
\usepackage{gensymb}
\usepackage{tabularx}
\usepackage{extarrows}
\usepackage{booktabs}
\usetikzlibrary{fadings}
\usetikzlibrary{patterns}
\usetikzlibrary{shadows.blur}

\usepackage[capitalise, noabbrev]{cleveref}

\usetikzlibrary{arrows, decorations.markings, shapes, calc}

\usetikzlibrary{shadows}
\tikzset{
diagonal fill/.style 2 args={fill=#2, path picture={
\fill[#1, sharp corners] (path picture bounding box.south west) -|
                         (path picture bounding box.north east) -- cycle;}},
reversed diagonal fill/.style 2 args={fill=#2, path picture={
\fill[#1, sharp corners] (path picture bounding box.north west) |- 
                         (path picture bounding box.south east) -- cycle;}}
}

\usetikzlibrary{hobby,backgrounds,calc,trees}

\usepackage{tikz-cd}
\usetikzlibrary{matrix}
\usetikzlibrary{babel}

\usepackage{caption}
\usepackage{subcaption}
\usepackage{float}
\restylefloat{table}

\usepackage{braket}

\usepackage[h]{esvect}

\usepackage{xurl} % xurl loads the url package and adds some stuff for additional breaking points for all alphanumerical characters and = / . : * - ~ ' "

\newtheorem{fact}[theorem]{Fact}

%% file: macros.tex
\newrobustcmd{\newnotion}[1]{\AP\intro{#1}}

\newcommand{\Self}{\mathsf{Self}}

\newcommand{\DL}[1]{\ensuremath{\mathcal{#1}}}  % Description Logics

\newcommand{\ALC}{\DL{ALC}}

\newcommand{\ALCO}{\DL{ALCO}}

\newcommand{\ALCQ}{\DL{ALCQ}}

\newcommand{\DLPhi}{\DL{L}\Phi}

\newcommand{\genericExt}{\theta}
\newcommand{\extb}{\textit{b}}
\newcommand{\extO}{\DL{O}}
\newcommand{\extI}{\DL{I}}

\newcommand{\extSelf}{\DL{\Self}}

\newcommand{\extList}{\extSelf, \extI, \extb, \extO}

\newcommand{\extSet}{\{ \extList \}}
\newcommand{\extTuple}{( \extList )}

\newcommand{\ALCOIbSelf}{\ALC_{\Self}\extI\extb\extO}
\newcommand{\ALCOIbSelfStdOrder}{\ALC\extO\extI\extb_{\Self}}
\newcommand{\ALCOSelfIb}{\ALC_{\Self}\extI\extb\extO}
\newcommand{\ALCOIb}{\ALC\extI\extb\extO}
\newcommand{\ALCOb}{\ALC\extO\extb}
\newcommand{\ALCSelf}{\ALC_{\Self}}

\newcommand{\role}[1]{\mathit{#1}}      % role name
\newcommand{\rolep}{\role{p}}           % role p
\newcommand{\roler}{\role{r}}           % role r
\newcommand{\rolerI}{\role{r}^{-}}      % role r-
\newcommand{\rolerIInv}{\role{r}^{-\interI}}      % role r-Int
\newcommand{\rolerJInv}{\role{r}^{-\interJ}}      % role r-Int
\newcommand{\roles}{\role{s}}           % role s
\newcommand{\concept}[1]{\mathrm{#1}}       % concept
\newcommand{\conceptA}{\concept{A}}         % concept A
\newcommand{\conceptB}{\concept{B}}         % concept B
\newcommand{\conceptC}{\concept{C}}         % concept C
\newcommand{\conceptD}{\concept{D}}         % concept D

\newcommand{\indv}[1]{\mathrm{#1}}  % individual name
\newcommand{\indvo}{\indv{o}}       % individual name o

\newcommand{\inter}[1]{\mathcal{#1}}    % interpretation
\newcommand{\interI}{\inter{I}}         % interpretation I
\newcommand{\interJ}{\inter{J}}         % interpretation J

\newcommand{\DeltaInter}[1]{\Delta^{#1}}                % delta of an interpretation
\newcommand{\DeltaI}{\DeltaInter{\interI}}              % delta of the interpretation I
\newcommand{\DeltaJ}{\DeltaInter{\interJ}}              % delta of the interpretation J

\newcommand{\cdotInter}[1]{\cdot^{#1}}          % cdot of an interpretation
\newcommand{\cdotI}{\cdotInter{\interI}}        % cdot of I
\newcommand{\el}[1]{\mathrm{#1}}                           % domain element
\newcommand{\elD}{\el{d}}                             % domain element d
\newcommand{\elE}{\el{e}}                             % domain element e

\newcommand{\lang}[1]{\mathbf{#1}}  % any "language" name
\newcommand{\Ilang}{\lang{N_I}}     % the set of all individual names
\newcommand{\Rlang}{\lang{N_R}}     % the set of all role names
\newcommand{\Clang}{\lang{N_C}}     % the set of all concept names
\newcommand{\cat}[1]{\mathbb{#1}}
\newcommand{\categoryC}{\cat{C}}
\newcommand{\categoryD}{\cat{D}}

\newcommand{\object}[1]{#1}
\newcommand{\objectA}{\object{A}}
\newcommand{\objectB}{\object{B}}
\newcommand{\objectC}{\object{C}}

\newcommand{\arrow}[1]{#1}
\newcommand{\arrowf}{\arrow{f}}
\newcommand{\arrowg}{\arrow{g}}

\newcommand{\arrowid}[1]{1_{#1}}

\newcommand{\comp}{\circ}

\newcommand{\functor}[1]{#1}

\newcommand{\functorG}{\functor{G}}

\newcommand{\sigmaO}{\sigma^\extO}

\newcommand{\sigmaI}{\sigma^\extI}

\newcommand{\sigmaSelf}{\sigma^\extSelf}
\newcommand{\sigmaTriple}{\sigma_i, \sigma_c, \sigma_r}

\newcommand{\unaryRel}[2]{#2 \in \conceptC^{#1}}
\newcommand{\struct}[1]{\inter{#1}}
\newcommand{\structA}{\struct{I}}
\newcommand{\structB}{\struct{J}}
\newcommand{\structPairBare}[2]{\struct{#1}, #2}
\newcommand{\structPair}[2]{(\structPairBare{#1}{#2})}
\newcommand{\structPairBareA}{\structPairBare{I}{d}}
\newcommand{\structPairBareB}{\structPairBare{J}{e}}
\newcommand{\structPairA}{(\structPairBareA)}
\newcommand{\structPairB}{(\structPairBareB)}
\newcommand{\ordStructA}{(\structA, d, \leq)}
\newcommand{\transPairA}[1]{(\trans{#1}(\interI), \elD)}
\newcommand{\transPairB}[1]{(\trans{#1}(\interJ), \elE)}

\newcommand{\bisimRelSym}[1]{\mathcal{Z}_{#1}}
\newcommand{\bisimRel}[1]{\bisimRelSym{#1}(d, e)}

\newcommand{\bisimRelOf}[3]{\bisimRelSym{#1}(#2, #3)}
\newcommand{\bisimRelAB}[1]{\bisimRelSym{#1}(a, b)}

\newcommand{\comonad}[1]{\functor{#1}}
\newcommand{\comonadG}{\comonad{G}}
\newcommand{\counit}{\varepsilon}
\newcommand{\counitOf}[1]{\counit_{#1}}
\newcommand{\counitA}{\counitOf{\structA}}
\newcommand{\counitB}{\counitOf{\structB}}
\newcommand{\coextensionof}[1]{ #1^{*} }

\newcommand{\comonadFunctorBare}{\mathbb{DL}}
\newcommand{\comonadFunctorOver}[1]{\comonadFunctorBare_k^{#1}}
\newcommand{\comonadFunctor}{\comonadFunctorOver{}}
\newcommand{\comonadFunctorPhi}{\comonadFunctorOver{\Phi}}

\newcommand{\comonadStructOver}[3]{\comonadFunctorOver{#1}\structPair{#2}{#3}}

\newcommand{\comonadStructA}{\comonadStructOver{}{I}{d}}
\newcommand{\comonadStructB}{\comonadStructOver{}{J}{e}}

\newcommand{\comonadStructProdOver}[1]{\comonadStructOver{#1}{I}{d} \times \comonadStructOver{#1}{J}{e}}
\newcommand{\comonadStructProdPhi}{\comonadStructProdOver{\Phi}}

\newcommand{\vocabV}{\mathcal{V}}
\newcommand{\vocabMap}{\delta}
\newcommand{\Nomega}{\mathbb{N} \cup \{ \omega \}}

\newcommand{\pointedCatOf}[1]{\mathcal{R}_* (#1)}
\newcommand{\pointedCat}{\pointedCatOf{\vocabV}}
\newcommand{\pointedCatALC}{\pointedCatOf{\emptyset, \sigma_c, \sigma_r}}

\newcommand{\pointedOrdCatK}{\mathcal{R}^D_{*k} (\vocabV)}
\newcommand{\EMcat}{EM(\comonadFunctorOver{})}
\newcommand{\pathsCat}{\textit{Paths}}
\newcommand{\morphismsOf}[1]{| #1 |}

\newcommand{\relBare}[2]{\roler_{\alpha_{#1}}^{#2}}
\newcommand{\rel}[4]{(#3, #4) \in \relBare{#1}{#2}}
\newcommand{\relA}[2]{\rel{}{\structA}{#1}{#2}}
\newcommand{\relB}[2]{\rel{}{\structB}{#1}{#2}}
\newcommand{\relComonad}[3]{\rel{#1}{\comonadStructA}{#2}{#3}}
\newcommand{\relComonadB}[3]{\rel{#1}{\comonadStructB}{#2}{#3}}

\newcommand{\bulletname}[1]{\textit{(#1)}}
\newcommand{\nameditem}[2]{\item{\bulletname{#1} \; #2}}
\newcommand{\embedding}{\rightarrowtail}

\newcommand{\JO}{J_\extO}
\newcommand{\JI}{J_\extI}
\newcommand{\Jb}{J_\extb}

\newcommand{\Jself}{J_\extSelf}

\newcommand{\trans}[1]{\mathfrak{f}_{#1}}
\newcommand{\transO}{\trans{\extO}}

\newcommand{\transI}{\trans{\extI}}
\newcommand{\transb}{\trans{\extb}}
\newcommand{\transSelf}{\trans{\extSelf}}

\newcommand{\backAndForthEq}{\leftrightarrow_k^{\comonadFunctorBare}}
\newcommand{\backAndForth}{\structPairA \backAndForthEq \structPairB}

\newcommand{\backAndForthGameSymbol}[1]{\mathcal{G}^{#1}_k}
\newcommand{\backAndForthGameOf}[1]{\backAndForthGameSymbol{#1}(\structPairBareA; \structPairBareB)}
\newcommand{\backAndForthGame}{\backAndForthGameOf{}}

\newcommand{\sequenceOf}[2]{[{#1}_0, \alpha_1, {#1}_1, ...,  \alpha_{#2}, {#1}_{#2} ]}
\newcommand{\sequence}{\sequenceOf{a}{j}}

\newcommand{\winningPos}{W(\structPairBareA ; \structPairBareB)}

\newcommand{\reduction}[1]{\mathfrak{#1}}
\newcommand{\reductionI}[1]{\reduction{#1}^{\interI}}
\newcommand{\reductionR}[1]{\reduction{#1}^{*}}

\newcommand{\KleisliCatOf}[1]{\mathsf{Kl}(#1)}

\newcommand{\kleisliHomo}{h : \comonadStructA \rightarrow \structPairB}
\newcommand{\kleisliExtHomo}{h^{*} : \comonadStructA \rightarrow \comonadStructB}

\newcommand{\compose}{\circ}

\newcommand{\redText}[1]{\textcolor[rgb]{0.82,0.01,0.11}{#1}}
\newcommand{\orangeText}[1]{\textcolor[rgb]{0.96,0.65,0.14}{#1}}

\makeatletter
\newcommand{\bigcomp}{%
  \DOTSB
  \mathop{\vphantom{\sum}\mathpalette\bigcomp@\relax}%
  \slimits@
}
\newcommand{\bigcomp@}[2]{%
  \begingroup\m@th
  \sbox\z@{$#1\sum$}%
  \setlength{\unitlength}{0.9\dimexpr\ht\z@+\dp\z@}%
  \vcenter{\hbox{%
    \begin{picture}(1,1)
    \bigcomp@linethickness{#1}
    \put(0.5,0.5){\circle{1}}
    \end{picture}%
  }}%
  \endgroup
}
\newcommand{\bigcomp@linethickness}[1]{%
  \linethickness{%
      \ifx#1\displaystyle 2\fontdimen8\textfont\else
      \ifx#1\textstyle 1.65\fontdimen8\textfont\else
      \ifx#1\scriptstyle 1.65\fontdimen8\scriptfont\else
      1.65\fontdimen8\scriptscriptfont\fi\fi\fi 3
  }%
}

\makeatother

%% file: 1-introduction.tex
\section{Introduction}\label{section:intro}

Following~\cite{Abramsky20}, there are two different views on the fundamental
features of computation, that can be summarised as ``structure'' and  ``power''
as follows:
\begin{itemize}\itemsep0em
  \item
    \textbf{Structure}: Compositionality and semantics, addressing the
    question of mastering the complexity of computer systems and taming
    computational effects.
  \item
    \textbf{Power}: Expressiveness and complexity, addressing the question of
    how we can harness the power of computation and recognize its limits.
\end{itemize}

It turned out that there are almost disjoint communities of researchers
studying Structure and Power, with seemingly no common technical language and
tools. To encounter this issue, Samson Abramsky and Anuj Dawar started a
project, whose goal is to provide a category-theoretical toolkit to reason about
finite model theory in order to apply theorems and draw insights from, at
first sight, an unrelated field.

Their approach employs comonads on the category of relational structures
to capture denotational semantics of
model comparison games such as Ehrenfeucht-Fraissé, pebbling, and bisimulation
games~\cite{relating-struct-and-pow}, as well as games for Hybrid
logics~\cite{comonads-hybrid-logic} and Guarded
Fragment~\cite{comonads-guarded-fragment}. The structure allows us to leverage
the tool of category theory, and apply it to generalise known established
theorems, as it was done in~\cite{lovasz-and-comonads} or \cite{arboreal-categories}.
In this paper, we continue the exploration of suitable game comonads by
incorporating the comonadic semantics for description logics games, namely, for
expressive description logics between $\ALC$ and $\ALCOSelfIb$. \footnote{It
will become clear why we write $\ALCOSelfIb$ instead of the more standard form
$\ALCOIbSelfStdOrder$ later.}
It is also worth mentioning parallel research that defines categorical semantics
for $\DL{ALC}$~\cite{category-semantics-for-DL, reasoning-in-DL-under-category-semantics}, however,
their approach is much different from ours, as we focus solely on games and leave \DL{ALC}
in the standard set-theoretic semantics.

\subsection{Our results}

In what follows, we change the setting established in the previous work~\cite{relating-struct-and-pow} from the
category of relational structures to a category of pointed interpretations
that are parametrised by sets of role names, concept names and individual
names.

We start by defining comonadic semantics for $\ALC$-bisimulation-games. It is
well-known that $\ALC$ is a notational variant of multi-modal logic. Hence,
we employ this observation to take full advantage of existing results on
modal logic from~\cite{relating-struct-and-pow} and use them as the base for
our further investigations. In order to define comonadic semantics for DLs
$\DLPhi \subseteq \ALCOSelfIb$, instead of providing it directly for them
(and thus repeating all the machinery and required proofs from~\cite{relating-struct-and-pow}),
we follow a different route. We provide a family of game reductions from $\DLPhi$
to weaker sublogics, ending up on $\ALC$, which transform interpretations in
such a way that a winning strategy in $\DLPhi$-bisimulation-game is equivalent
to a winning strategy in $\ALC$-bisimulation-game for suitably transformed
interpretations.
From a categorical point of view, we introduce a comonad for $\ALC$ logic and
reductions shall be defined by functors, on which we will build relative
comonads to encapsulate the additional capabilities available in an
$\DL{L}$-bisimulation-game. By composing the reduction functors together, we
shall obtain comonadic semantics for all of the games for considered logics.

% TODO: add some philosophical accent on what happens

\subsection{Roadmap}

We start in \cref{section:preliminaries} by giving a sufficient background
for the further results, to facilitate the accessibility for readers coming
both from the area of model theory and description logics, as well as from the
category theory side.

In~\cref{section:bisim-games}, we recall the well-established notion of
bisimulation games for $\DL{L} \subseteq \ALCOSelfIb$ logics, which are the key
concept for which we shall define the categorical semantics.

We then proceed to \cref{section:game-reductions}, where we define a family of
logic extension reductions
$\trans{\Self}, \trans{\extI}, \trans{\extb}, and \trans{\extO}$
acting on interpretations. We declare them with a goal such that for $\Phi \subseteq \extSet$ and $\trans{\Phi}$
being a composition of reductions of extensions selected by $\Phi$, the
following theorem holds:

\begin{center}
            $(\interI, \elD) \sim_{k}^{\ALC\Phi} (\interJ, \elE) \;
            \iff
            \: (\trans{\Phi} \: \interI, \elD) \sim_{k}^{\ALC} (\trans{\Phi} \: \interJ, \elE)$
\end{center}

Having established model-theoretic part of our work, we finally move to the
category theory world, where we shall stay until for the rest of the thesis.
\cref{section:game-comonads} tweaks modal comonad and ports categorical
variation of comparison games from~\cite{relating-struct-and-pow} such that it
can be applied to our description logic setting. We wrap up the chapter by
giving denotational, comonadic semantics for  $\ALC$-bisimulation-games.

Finally, in~\cref{section:extensions}, we devise a general framework for
establishing comonadic semantics for games for all expressive sublogics
of $\ALCOIbSelf$. We achieve this by lifting previously defined reductions to
well-behaved functors and taking a relative comonad over them.

We conclude in~\cref{section:conclusions} by suggesting potential future research
directions as well as giving motivation to the thesis by presenting what was already
achieved in this field by leveraging the developed toolkit.

%% file: 2-preliminaries.tex
\section{Preliminaries}\label{section:preliminaries}

We start with a recap of notions from category theory~\cite{awodeybook,maclane71},
such as comonads, as well as from description logics, for which we define their
syntax, semantics and bisimulations~\cite{dlbook}. By doing so, we would like
to unify the context for readers from different backgrounds.

\subsection{Preliminaries on DLs.}

We fix infinite mutually disjoint sets of \emph{individual names} \(\Ilang \),
\emph{concept names} \(\Clang \), and \emph{role names}~\( \Rlang \). We will
briefly recap the syntax and semantics of $\ALCOSelfIb$-concepts and as well as
$\DL{L}$-concepts for relevant sublogics $\DL{L}$ of $\ALCOSelfIb$. The
following EBNF grammar defines \emph{atomic concepts} $\conceptB$,
\emph{concepts} $\conceptC$, \emph{atomic roles} $\roler$,
\emph{simple roles}~$\roles$ with $\indvo \in \Ilang$, $\conceptA \in \Clang$, $\rolep \in \Rlang$:

$$
\begin{array}{rllrll}
    \conceptB & ::= & \conceptA \mid \{ \indvo \} \\
    \conceptC & ::= & \conceptB
                      \mid \neg \conceptC
                      \mid \conceptC \sqcap \conceptC
                      \mid \exists \roles.\conceptC
                      \mid \exists \roles.\Self \\

    \roler & ::= & \rolep \mid \rolep^- \\
    \roles & ::= & \roler
                   \mid \roles \cap \roles
                   \mid \roles \cup \roles
                   \mid \roles \setminus \roles \\
\end{array}
$$

The semantics of \( \ALCOSelfIb \)-concepts is defined via \emph{interpretations}
\(\interI = (\DeltaI, \cdotI) \) composed of a non-empty set \(\DeltaI \)
called the \emph{domain of \(\interI \)} and an \emph{interpretation function}
\(\cdotI \) mapping individual names to elements of \(\DeltaI \), concept names
to subsets of \(\DeltaI \), and role names to subsets of \(\DeltaI \times \DeltaI \).
This mapping is then extended to complex concepts and roles
(\cf~\cref{tab:concepts-and-roles}). The \emph{rank} of a concept is the
maximal nesting depth of $\exists$-restrictions.

We shall use expressions of the form $\ALC\Phi$ or $\DL{L}\Phi$ with
$\Phi \subseteq \{ \extO, \extI, \extSelf, \extb\}$ to speak
collectively about different expressive sublogics of $\ALCOSelfIb$.

The $\ALC\Phi$-concepts are obtained by dropping from the syntax
the inversions of roles $(\mathcal{I})$, safe boolean combination of roles
$(\mathit{b})$ (\ie role union, intersection and difference), nominals
$(\mathcal{O})$  and the self operator $(\Self)$, depending on the content of
$\Phi$. We stress here that role union/intersection/difference, the $\Self$
operator, role inverse $\cdot^-$ and nominals $\{ \cdot \}$ are just operators
and they introduce neither new role names nor new concept names.

\begin{table}[H]
      \caption{Concepts and roles in \( \ALCOSelfIb \).\label{tab:concepts-and-roles}}
      \centering
        \begin{tabular}{@{}l@{\ \ \ }c@{\ \ \ }l@{}}
            \hline\\[-2ex]
            Name & Syntax & Semantics \\ \hline \\[-2ex]
            concept name
                & \(\conceptA \) & \(\conceptA^\interI \subseteq \DeltaI  \) \\
            role name
                & \(\roler \) & \(\roler^\interI \subseteq \DeltaI {\times} \DeltaI \) \\
            concept negation
                & \(\neg\conceptC \)& \(\DeltaI \setminus \conceptC^{\interI} \) \\  %\hline
            concept intersection
                & \(\conceptC \sqcap \conceptD \)& \(\conceptC^{\interI}\cap \conceptD^{\interI} \) \\  %\hline
            existential restriction
                & \(\exists{\roler}.\conceptC \) &
                \(\set{
                    \elD | \exists{\elE}.(\elD,\elE)
                    \in \roler^{\interI} \land \elE\in \conceptC^{\interI}
                } \)\\
            nominal op.
                & $\{\indvo\}$ & $\{\indvo^{\interI}\}$ \\
            inverse role op.
                & $\rolep^-$ & $\{
                    (\elD, \elE) \mid (\elE, \elD) \in \rolep^\interI
                \}$ \\
            role boolean op. for $\oplus \in \{ \cup, \cap, \setminus \}$
                & $\roles_1 \oplus \roles_2$ & $\roles_1^\interI \oplus \roles_2^\interI$\\
            $\Self$ op.
                & $\exists \roles.\Self$ & $\{\elD \mid (\elD, \elD) \in \roles^{\interI}\}$\\
        \end{tabular}
\end{table}

Any triple $\vocabV \triangleq (\sigma_i, \sigma_c, \sigma_r)$ from
$\Ilang \times \Clang \times \Rlang$ having finite components will be called a
\emph{vocabulary}. We say $\DL{L}(\vocabV)$-concepts for those
$\DL{L}$-concepts that employ only symbols from $\vocabV$. For a \emph{pointed
interpretation} $(\interI, \elD)$ we say that it \emph{satisfies} a
concept $\conceptC$ (written: $(\interI, \elD) \models \conceptC)$ if
$\elD \in \conceptC^{\interI}$. A $\vocabV$-pointed-interpretation
$(\interI, \elD)$ is a partial interpretation, where all individual names
outside $\vocabV$ are left undefined while other symbols outside
$\vocabV$ are interpreted~as~$\emptyset$.

\subsection{Preliminaries on category theory}

We assume familiarity with basic concepts such as categories, functors or natural transformations.
For a definition of a category, functor and natural transformation, see~
\cite[Deﬁnition 1.1, 1.2 and 7.6]{awodeybook}. Let $\categoryC$ and
$\categoryD$ be categories. We write $\morphismsOf{\categoryC}$ to denote
morphisms (arrows) of $\categoryC$ and $f \in \morphismsOf{\categoryC}$ to
indicate that $f$ is a morphism in $\categoryC$.

%
% Comonads
%
Let $\functorG: \categoryC \to \categoryC$ be a functor and
$\counit : \categoryC \Rightarrow 1_\categoryC$ a natural transformation, with
$1_\categoryC$ being the identity functor on $\categoryC$.

\begin{definition}

A \emph{comonad} $\comonadG$ is a triple $(\functorG, \counit,
\coextensionof{(\cdot)})$, where~$\counit$ is called the \emph{counit} of
$\comonadG$ that for each object $\objectA$ it gives us an arrow
$\counitOf{\objectA}: \comonadG\objectA \to \objectA$, while
$\coextensionof{(\cdot)}$, called the \emph{Kleisli coextension} of
$\comonadG$, is an operator sending each arrow $\arrowf : \comonadG\objectA \to
\objectB$ to $\coextensionof{\arrowf}: \comonadG\objectA \to
\comonadG\objectB$.

These have
to satisfy, for all $\arrowf: \comonadG\objectA \to \objectB$ and
$\arrowg: \comonadG\objectB \to \objectC$, the equations:
\[
    \coextensionof{\counitOf{\objectA}} = \arrowid{\comonadG\objectA}, \qquad\qquad\qquad 
    \counitOf{\objectB} \comp \coextensionof{\arrowf} = \arrowf, \qquad\qquad\qquad 
    \coextensionof{(\arrowg \comp \coextensionof{\arrowf})} = \coextensionof{\arrowg} \comp \coextensionof{\arrowf}
\]

 \begin{minipage}{\linewidth}
  \centering
  \begin{minipage}{0.35\linewidth}

    \begin{center}
    \begin{tikzcd}[row sep=huge, column sep=huge]
    {\comonadStructA} \arrow[d, "f^*"'] \arrow[rd, "f"] &    \\
    {\comonadStructB} \arrow[r, "\counitA"']            & {\structPairB}
    \end{tikzcd}
    \end{center}

  \end{minipage}
  \hspace{0.05\linewidth}
  \begin{minipage}{0.55\linewidth}

    \begin{center}
    \begin{tikzcd}[row sep=huge, column sep=huge]
        \comonadStructA \arrow[d, "f^*"'] \arrow[rd, "(g \comp f^*)^*"] &   \\
        \comonadStructB \arrow[r, "g^*"']            & \comonadStructOver{}{K}{k}
    \end{tikzcd}
    \end{center}

    \end{minipage}
  \end{minipage} \\

\end{definition}

\begin{definition}

    A \emph{coKleisli category} $\KleisliCatOf{\comonadG}$ is a category with objects
    from $\categoryC$ and arrows from $\objectA$ to $\objectB$ given by the arrows
    in $\categoryC$ of the form~$\comonadG\objectA \to \objectB$, where composition
    $g \bullet f$ is given by $g \comp \coextensionof{f}$.

\end{definition}

% TODO: add diagrams for comonads

%
% Relative comonads
%
We shall also need the notion of \emph{relative comonads}~\cite{monads-need-not-be-endofunctors}:

\begin{definition}[Relative comonad]

Given a functor $ J : \categoryC \rightarrow \categoryD $, and a comonad $G$ on
$\categoryD$, we obtain a \textit{relative comonad} on $\categoryC$, whose
coKleisli category is defined as follows. A morphism from $A$ to $B$, for
objects $A$, $B$ of $\categoryC$, is a $\categoryD$-arrow $GJA \rightarrow JB$.
The counit at $A$ is $\counitOf{JA}$, using the counit of $G$ at $JA$. Given
$f : GJA \rightarrow JB$, the Kleisli coextension $f^{*} : GJA \rightarrow GJB$
is the Kleisli coextension of G. Since G is a comonad, these operations satisfy
the equations for a comonad in Kleisli form. We write this as $(G \circ J)$-relative-comonad.

\end{definition}

% TODO: make this more verbose, as in N. Shah thesis

%% file: 3-bisimulation-games.tex
%!TEX root = main.tex

\section{Bisimulation Games}\label{section:bisim-games}

% TODO: some text here
We now shall recall the characterization of the equality of interpretations under
a certain logic via bisimulation games and bisimulation relation and argue their
logical equivalence.

\begin{definition}
We write $(\interI, \elD) \equiv_k^{\DL{L}\Phi(\vocabV)} (\interJ,
\elE)$ iff $\elD$ and $\elE$ satisfy the same
$\DL{L}\Phi(\vocabV)$-concepts of rank at most $k$, where $k \in \Nomega$.
\end{definition}

\subsection{Games}

Let $\vocabV$ be a vocabulary. Following~\cite{Piro12}, we recap the notion of
\emph{bisimulation games} for $\ALC$ and its extensions.

\begin{definition}
    Call $\elD \in \Delta^{\interI}$ and $\elE \in \Delta^{\interJ}$ to be in
    $\vocabV$-\emph{harmony}\footnote{For $\ALC$ we do not actually use $\sigma_i$
    and $\sigma_r$, but they will be useful for other logics.} if for all concept names
    $\conceptC \in \sigma_c$ we have that $\elD \in \conceptC^{\interI}$ iff
    $\elE \in \conceptC^{\interJ}$.
\end{definition}

The $\ALC(\vocabV)$-\emph{bisimulation
game} is played by two players, Spoiler (he) and Duplicator (she), on two pointed
interpretations $(\interI, \elD_0)$ and
$(\interJ, \elE_0)$. A~\emph{configuration} of a game is a quartet of the
form $(\interI, s; \interJ, s')$, where $s$ and
$s'$ are words from, respectively, $\DeltaI (\sigma_r \DeltaI)^* $ and
$\DeltaJ (\sigma_r \DeltaJ)^*$. Intuitively, configurations encode not only
the current position of
the play but also its full play history. The \emph{initial configuration} is
simply $(\interI, \elD_0; \interJ, \elE_0)$. The $0$-th round of the
game starts in the initial configuration and we require that $\elD_0$ and
$\elE_0$ are in $\vocabV$-harmony. If not, then immediately Spoiler wins.
For any configuration $(\interI, s\elD ; \interJ, s'\elE)$ (where the
sequences $s,s'$ may be empty) in the game, the following rules apply:

\begin{enumerate}[(a)]
  \item
        In each round, Spoiler picks one of the two interpretations, say
        $\interI$. Then he picks a role name $\roler \in \sigma_r$ and takes
        an element $\elD' \in \DeltaI$ such that ($\heartsuit$):
        $(\elD,\elD') \in \roler^{\interI}$. If there is no such role
        name $\roler$ and an element $\elD'$, then Duplicator wins.

  \item
        Duplicator responds in the other interpretation, $\interJ$, by
        picking the same role name $\roler \in \sigma_r$ as Spoiler did and an
        element $\elE' \in \DeltaI$ in $\vocabV$-harmony with
        $\elD'$, witnessing ($\clubsuit$): $(\elE,\elE') \in
        \roler^{\interJ}$. If there is no such role name $\roler$ or an element
        $\elE'$, Spoiler wins.

\end{enumerate}

The game continues from the position
$(\interI, s\elD\roler\elD' ; \interJ, s'\elE\roler\elE')$.
Duplicator has a winning strategy in the game on $(\interI, \elD_0; \interJ,
\elE_0)$ if she can respond to every move of Spoiler so that she either wins
the game or can survive $\omega$ rounds. We define winning strategies in
$k$-round games analogously.

The above game is adjusted to the case of expressive sublogics $\DL{L}\Phi$ of
$\ALCOIbSelf$~as~follows.

\begin{itemize}\itemsep0em
  \item If $\extO \in \Phi$, then we extend the definition
        of $\vocabV$-harmony with a condition ``for all
        $\indvo \in \sigma_i$ we have that $\elD = \indvo^{\interI}$ iff
        $\elE = \indvo^{\interJ}$''.
  \item If $\Self \in \Phi$, then we extend the definition
        of $\vocabV$-harmony with a condition ``for all $\roler \in \sigma_r$ we
        have that $(\elD, \elD) \in \roler^{\interI}$ iff
        $(\elE, \elE) \in \roler^{\interJ}$''.
  \item If $\extI \in \Phi$, then in Spoiler's move the
        condition $(\heartsuit)$ additionally allows for $(\elD',\elD)
        \in \roler^{\interI}$. Then in the corresponding move of Duplicator,
        the condition $(\clubsuit)$ imposes $(\elE',\elE) \in
        \roler^{\interJ}$.
  \item If $\extb \in \Phi$, then for the element $\elE'$ we additionally
        extend $(\clubsuit)$ to fulfil the equality
        $\{
            \roler \in \sigma_r \mid (\elD, \elD') \in
            \roler^{\interI}
         \} = \{
            \roler \in \sigma_r \mid (\elE,
            \elE') \in \roler^{\interJ}
        \}$.
        Moreover, in case of $\extI \in \Phi$ then also
        $ \{
            \roler \in \sigma_r \mid (\elD', \elD) \in
            \roler^{\interI}
        \} = \{
            \roler \in \sigma_r \mid (\elE',
            \elE) \in \roler^{\interJ}
        \} $ must hold.

\end{itemize}

\begin{proposition}
    $\vocabV$-\emph{harmony} is a transitive relation under all game variations
\end{proposition}
\begin{proof}
    Notice that in the definition we have used everywhere logical equivalence,
    from which transitivity follows directly. Clearly combining logics together preserves that.
\end{proof}

The following fact for most of the considered logics is either well-known
(see~\cite{Piro12}, in particular, Prop. 2.1.3 and related chapters) or can be
established by tiny modifications of the existing proofs:

\begin{fact}\label{fact:dupl-has-winstr-iff-sat-the-same-alc-conc}
    Let $\DL{L}$ be a description logic satisfying $\ALC \subseteq \DL{L} \subseteq \ALCOIbSelf$.
    Duplicator has a winning strategy in $\DL{L}(\vocabV)$-bisimulation game
    played on two pointed interpretations $(\interI, \elD)$ and $(\interJ, \elE)$
    iff $(\interI, \elD)$ and $(\interJ, \elE)$ satisfy the same $\DL{L}(\vocabV)$-concepts.
\end{fact}

\subsection{Bisimulations}

To simplify reasoning about bisimulation games, we employ the well-known notion
of \emph{bisimulation}, which can be seen as the ``encoding'' of winning
strategies of Duplicator.
Let $\DL{L}\Phi$ be an expressive sublogic of $\ALCOSelfIb$ and
$k \in \Nomega$.
Following~\cite{DivroodiN15}:

\begin{definition}[Bisimulation relation]\label{def:bisim-relation}

$\DL{L}\Phi(\vocabV)$-$k$-\emph{bisimulation} between $(\interI, \mathrm{a})$
and $(\interJ, \mathrm{b})$ is a set $\mathcal{Z} \subseteq
\bigcup_{\ell=0}^{k} (\DeltaI)^k \times (\DeltaJ)^k$ satisfying the following
seven conditions for all $\indvo \in \sigma_i, \conceptC \in \sigma_c, \roler
\in \sigma_r, \elD,\elD' \in \DeltaI$, $s \in (\DeltaI)^*$ and
$\elE, \elE' \in \DeltaJ, s' \in (\DeltaJ)^*$:

\begin{enumerate}[(a)]\itemsep0em
  \item If $\mathcal{Z}(s\elD, s'\elE)$ then $\elD \in
      \conceptC^{\interI}$ iff $\elE \in \conceptC^{\interJ}$.

  \item If $\mathcal{Z}(s\elD, s'\elE)$ and $(\elD, \elD') \in
      \roler^{\interI}$ then there is $\elE' \in \DeltaJ$ s.t.
        $(\elE, \elE') \in \roler^{\interJ}$ and
        $\mathcal{Z}(s\elD\elD', s'\elE\elE')$.

  \item If $\mathcal{Z}(s\elD, s'\elE)$ and $(\elE, \elE') \in
      \roler^{\interJ}$ then there is $\elD' \in \DeltaI$ s.t.
        $(\elD, \elD') \in \roler^{\interI}$ and
        $\mathcal{Z}(s\elD\elD', s'\elE\elE')$.

  \item If $\extO \in \Phi$, then $\mathcal{Z}(s\elD, s'\elE)$
        implies $\elD = \indvo^{\interI}$ iff $\elE = \indvo^{\interJ}$.

  \item If $\Self \in \Phi$, then $\mathcal{Z}(s\elD, s'\elE)$
        implies $(\elD, \elD) \in \roler^{\interI}$ iff
        $(\elE, \elE) \in \roler^{\interJ}$.

  \item If $\extI \in \Phi$, then $\mathcal{Z}(s\elD, s'\elE)$
      and $(\elD', \elD) \in \roler^{\interI}$ implies
        that there is $\elE' \in \DeltaJ$ such that
        $(\elE', \elE) \in \roler^{\interJ}$ and $\mathcal{Z}(s\elD\elD', s'\elE\elE')$.

  \item If $\extb \in \Phi$,, then if
        $\mathcal{Z}(s\elD, s'\elE)$ and
        $(\elD, \elD') \in \roler^{\interI}$ implies that there is
        $\elE' \in \DeltaJ$ satisfying $\mathcal{Z}(s\elD\elD', s'\elE\elE')$
        and $
        \{ \roler \in \sigma_r \mid (\elD, \elD') \in \roler^{\interI} \} =
        \{ \roler \in \sigma_r \mid (\elE, \elE') \in \roler^{\interJ} \}$
        . If $\extI \in \Phi$, then also
        $\{ \roler \in \sigma_r \mid
        (\elD', \elD) \in \roler^{\interI} \} = \{ \roler \in
        \sigma_r \mid (\elE', \elE) \in \roler^{\interJ} \}$.
\end{enumerate}

\end{definition}

% The notion of $k$-bisimulation is lifted to the case of
% $\omega$-bisimulations, by allowing bisimilar sequences to to be arbitrarily long.

Note that if $\mathcal{Z}$ is an $\omega$-bisimulation, then $\mathcal{Z}$
becomes a $k$-bisimulation when restricted to pairs of sequences of length at
most $k$.
% Again, it is straightforward to show: %(e.g. by following Prop.
2.1.3 from~\cite{Piro12}) that: The following fact for most of the considered
logics is either well-known (see~\cite{Piro12}, in particular, Prop. 2.1.3 and
related chapters) or can be established by tiny modifications of existing
proofs.

\begin{fact}\label{fact:games-bisimulations-and-concept-eq}
    For any $k \in \mathbb{N} \cup \{ \omega \}$ and a logic $\DL{L}\Phi$ between
    $\ALC$ and $\ALCOIbSelf$, t.f.a.e.:

    \begin{itemize}
        \item Duplicator has the winning strategy in the $k$-round $\DL{L}\Phi(\vocabV)$-bisimulation-game
             on~$(\interI, \elD ; \interJ, \elE)$,

        \item There is an $\DL{L}\Phi(\vocabV)$-$k$-bisimulation
            $\mathcal{Z}$ between $\structPairA$ and $\structPairB$ such that $\mathcal{Z}(\elD, \elE)$,

        \item $(\interI, \elD) \equiv_k^{\DL{L}\Phi(\vocabV)} (\interJ, \elE)$.
    \end{itemize}
\end{fact}

%% file: 4-reductions.tex
%!TEX root = main.tex

\section{Reductions between games and logics}
\label{section:game-reductions}

Herein we establish reductions, based on appropriate model transformations,
that will allow us to transfer the winning strategies of Duplicator from
richer logics to weaker ones, ending up on $\ALC$. All of them, except the case
of nominals, will be trivial. Such transformation will be essential
in~\cref{section:extensions}, where we shall employ them in the
construction~of~relative~comonads.

We will denote the game reductions for logic extensions $\Phi \subseteq \extSet$ by $\reduction{f}_\Phi$,
which has two components $\reductionI{f}_\Phi$ and $\reductionR{f}_\Phi$, that define
actions on, respectively, the interpretation and the distinguished element.

\begin{definition}
    Let $\interI$ be an interpretation over vocabulary $(\sigma_i, \sigma_c, \sigma_r)$.
    A \textit{$(\sigma_i', \sigma_c', \sigma_r')$-reduct} of an interpretation
    $\interI$ is an interpretation $\interI'$ obtained
    by interpreting all the symbols outside of $\sigma_i' \cup \sigma_c' \cup \sigma_r'$
    as empty sets.
\end{definition}

\subsection{$\Self$ operator}

We first handle the $\Self$ operator.
Let $\sigmaSelf_c \triangleq \sigma_c \cup \{\conceptC_{\Self.\roler} \mid \roler \in \sigma_r\} $.
By the \emph{self-enrichment} of a $\vocabV \triangleq (\sigma_i, \sigma_c, \sigma_r)$-interpretation
$\interI$ we mean the $\vocabV^{\extSelf} \triangleq (\sigma_i, \sigmaSelf_c, \sigma_r)$-interpretation $\interI_{\Self}$,
where the $(\sigma_i, \sigma_c, \sigma_r)$-reduct
of $\interI_{\Self}$ is equal to $\interI$ and the
interpretations of $\conceptC_{\Self.\roler}$ concepts are defined as
$(\conceptC_{\Self.\roler})^{\interI_{\Self}} = (\exists\roler.\Self)^{\interI}$.

\begin{center}

    \input{graphs/self_reduction.tex}

\end{center}

Let $\trans{\Self}$ be the described transformation, mapping $(\interI, \elD)$ to~$(\interI_{\Self}, \elD)$.

\begin{proposition}\label{prop:bISelf-game-reduction}
    Let $k \in \Nomega$ and let $\DL{L}$ be a DL satisfying $\ALC \subseteq
    \DL{L} \subseteq \ALCOIb$. Then Duplicator has a winning strategy in a $k$-round
    $\DL{L}_\Self(\vocabV)$-bisimulation game on
    $(\interI, \elD ; \interJ, \elE)$ iff she has a winning strategy
    in a $k$-round $\DL{L}(\vocabV)$-bisimulation game on
    $(\trans{\Self}(\interI), \elD ; \trans{\Self}(\interJ), \elE)$.

\end{proposition}
\begin{proof}
    By applying~\cref{fact:games-bisimulations-and-concept-eq} to both sides,
    it is sufficient to prove the following:

    \textit{
        There is a $\DL{L}_\Self(\vocabV)$-$k$-bisimulation
        $\bisimRelSym{}$ between $\structPairA$ and $\structPairB$ such that $\bisimRel{}$
        iff
        there is a $\DL{L}(\vocabV^{\extSelf})$-$k$-bisimulation
        $\bisimRelSym{\Self}$ between $\transPairA{\Self}$ and $\transPairB{\Self}$ such that $\bisimRel{\Self}$
    }

\noindent($\Longrightarrow$) Let us assume $\bisimRelSym{}$ is the bisimulation from
implication predecessor and take $\bisimRelSym{\Self} \triangleq \bisimRelSym{}$.
We now need to prove that $\bisimRelSym{\Self}$ is a valid bisimulation. Notice
that the only way in which $\trans{\Self}$-reduced interpretations differ are the
atomic concepts, so it is sufficient to prove that case $(a)$
from~\cref{def:bisim-relation} holds for new $\conceptC_{\Self.\roler}$ concepts.
Take any $a \in \interI, b \in \interJ$.

    \begin{align*}
        \bisimRelAB{\Self} &\Longrightarrow \bisimRelAB{}       && \bisimRelSym{\Self} = \bisimRelSym{} \\
        &\Longrightarrow (a, a) \in \roler^{\interI} \iff (b, b) \in \roler^{\interJ}       && \text{def. } \bisimRelSym{}, \: (e) \\
        &\Longrightarrow a \in (\exists\roler.\Self)^{\interI} \iff b \in (\exists\roler.\Self)^{\interJ}     && \text{def. } \exists\roler.\Self \\
        &\Longrightarrow a \in (\conceptC_{\Self.\roler})^{\interI_{\Self}} \iff b \in (\conceptC_{\Self.\roler})^{\interJ_{\Self}} && \text{def. } \conceptC_{\Self.\roler} \\
    \end{align*}

\noindent($\Longleftarrow$) Proof for the other side is analogous. Let us again
assume $\bisimRelSym{\Self}$ is the bisimulation from implication predecessor
and take $\bisimRelSym{} \triangleq \bisimRelSym{\Self}$. We now need to prove
that $\bisimRelSym{}$ is a valid bisimulation. This time, the only case that
needs special attention is $(e)$ from~\cref{def:bisim-relation}.
Take any $a \in \interI, b \in \interJ$.

    \begin{align*}
        \bisimRelAB{} &\Longrightarrow \bisimRelAB{\Self}       && \bisimRelSym{} = \bisimRelSym{\Self} \\
        &\Longrightarrow a \in (\conceptC_{\Self.\roler})^{\interI_{\Self}} \iff b \in (\conceptC_{\Self.\roler})^{\interJ_{\Self}}     && \text{def. } \bisimRelSym{}, \: (a) \\
        &\Longrightarrow a \in (\exists\roler.\Self)^{\interI} \iff b \in (\exists\roler.\Self)^{\interJ}     && \text{def. } \conceptC_{\Self.\roler} \\
        &\Longrightarrow (a, a) \in \roler^{\interI} \iff (b, b) \in \roler^{\interJ}       && \text{def. } \exists\roler.\Self \\
    \end{align*}

\end{proof}

\subsection{Role inverses}

Our next goal is to incorporate inverses of roles.
Let $\sigmaI_r \triangleq \sigma_r \cup \{ \roler_{\textit{inv}} \mid \roler \in \sigma_r \}$
By the \emph{inverse-enrichment} of a $\vocabV \triangleq (\sigma_i, \sigma_c, \sigma_r)$-interpretation
$\interI$ we mean the $\vocabV^{\extI} \triangleq (\sigma_i, \sigma_c, \sigmaI_r)$-interpretation $\interI_{\extI}$,
where the $(\sigma_i, \sigma_c, \emptyset)$-reducts of $\interI$ and $\interI_{\extI}$ are equal, and the
interpretations of role names $\roler_{\textit{inv}}$ are defined as
$(\roler_{\textit{inv}})^{\interI_{\extI}} = (\roler^-)^{\interI}$.

\begin{center}

\input{graphs/inverses_reduction.tex}

\end{center}

Let $\trans{\extI}$ be the described transformation, mapping $(\interI, \elD)$ to
$(\interI_{\extI}, \elD)$.
The proposition follows in a similar pattern to~\cref{prop:bISelf-game-reduction}:

\begin{proposition}\label{prop:bI-game-reduction}
    Let $k \in \Nomega$ and let $\DL{L}$ be a DL satisfying $\ALC \subseteq \DL{L} \subseteq \ALCOb$.
    Then Duplicator has a winning strategy in a $k$-round $\DL{L}\extI(\vocabV)$-bisimulation game
    on $(\interI, \elD ; \interJ, \elE)$
    iff she has a winning strategy in a $k$-round $\DL{L}(\vocabV^{\extI})$-bisimulation game on
    $(\trans{\extI}(\interI), \elD; \trans{\extI}(\interJ), \elE)$.
\end{proposition}

\begin{proof}
    By applying~\cref{fact:games-bisimulations-and-concept-eq} to both sides,
    it is sufficient to prove the following:

    \textit{
        There is a $\DL{L}_\extI(\vocabV)$-$k$-bisimulation
        $\bisimRelSym{}$ between $\structPairA$ and $\structPairB$ such that $\bisimRel{}$
        iff
        there is a $\DL{L}(\vocabV^{\extI})$-$k$-bisimulation
        $\bisimRelSym{\extI}$ between $\transPairA{\extI}$ and $\transPairB{\extI}$ such that $\bisimRel{\extI}$
    }

\noindent($\Longrightarrow$) Let us assume $\bisimRelSym{}$ is the bisimulation from
implication predecessor and take $\bisimRelSym{\extI} \triangleq \bisimRelSym{}$.
Notice that the only way in which $\trans{\extI}$-reduced interpretations differ
are the added fresh inverse roles, so it is sufficient to prove that cases
$(b)$ and $(c)$ from~\cref{def:bisim-relation} hold for $\sigma_r^{\extI}$ roles.
The case for roles in $\sigma_r$ is trivial, as there were no changes to them made and we have
that $\bisimRelSym{\extI} = \bisimRelSym{}$. Take
$a, a' \in \interI, b \in \interJ, \rolerI \in \sigma_r^{\extI} \setminus \sigma_r$
and assume that $\bisimRelAB{\extI}$ and $(a', a) \in \rolerIInv$. Let us consider the case $(b)$,
case $(c)$ will follow analogously.
We need to show that there exists $b' \in \interJ$ s.t. $(b', b) \in \rolerJInv$
and $\bisimRelOf{\extI}{aa'}{bb'}$. By construction, $\rolerI$ has a corresponding role
$\roler$ s.t. $(a, a') \in \roler^{\interI}$. From $\bisimRelAB{}$ assumption,
we can extract $b'$ s.t. $(b, b') \in \roler^{\interJ}$. By definition of the construction,
this implies that $(b', b) \in \rolerJInv$ which closes the proof.

\noindent($\Longleftarrow$) Proceeds similarly as the proof above.

\end{proof}

\subsection{Safe boolean roles combinations}

We focus next on safe boolean combinations of roles. Given a finite $\sigma_r
\subseteq \Rlang$, let $\sigma_r^{\extb}$ be composed of role names having the
form $\roler_S$, where $S$ is any non-empty subset of $\sigma_r$. By the
\emph{b-enrichment} of a $\vocabV \triangleq (\sigma_i, \sigma_c,
\sigma_r)$-interpretation $\interI$ we mean the $\vocabV^{\extb} \triangleq (\sigma_i,
\sigma_c, \sigma_r^{\extb})$-interpretation $\interI_{\extb}$, where the
$(\sigma_i, \sigma_c, \emptyset)$-reducts of $\interI$ and $\interI_{\extb}$
are equal and the interpretation of role names $\roler_S \in \sigma_r^{\extb}$
is defined as
$
    \{
        (\elD, \elE)
        \mid
        S = \{ \roler \in \sigma_r \mid (\elD, \elE) \in \roler^{\interI} \}
    \}
$.

\begin{center}

\input{graphs/bool_comb_reduction.tex}

\end{center}

Let $\trans{\extb}$ be the
described transformation, mapping $(\interI, \elD)$ to $(\interI_{\extb}, \elD)$. Once more,
the following proposition is straightforward:

\begin{proposition}\label{prop:b-game-reduction}
    Let $k \in \Nomega$ and let $\DL{L}$ be a DL satisfying $\ALC \subseteq \DL{L} \subseteq \ALCO$.
    Then Duplicator has a winning strategy in a $k$-round
    $\DL{L}\mathit{b}(\vocabV)$-bisimulation-game on $(\interI, \elD ;
    \interJ, \elE)$ iff she has a winning strategy in a $k$-round
    $\DL{L}(\vocabV^{\extb})$-bisimulation-game on $(\trans{\extb}(\interI),
    \elD ; \trans{\extb}(\interJ), \elE)$.
\end{proposition}
\begin{proof}

    The key observation here is that safe boolean roles combinations are
    giving us the power to define any 2-type as a step in the bisimulation.
    Henceforth, we convert the interpretation such that the arrows represent exactly
    2-types and therefore a move in the game can cover any move that
    could have been expressed by roles combinations. A detailed proof is
    very similar to~\cref{prop:bISelf-game-reduction} and \cref{prop:bI-game-reduction}
    and thus shall be left as an exercise for the reader.
\end{proof}

\subsection{Nominals}

Finally, we proceed with the case of nominals. In this case, we need to be extra
careful, as the comonads introduced in the next section will act as unravelling
on interpretations, and we do not want to create multiple copies of a nominal.
Recall that the Gaifman graph $\mathsf{G}_{\interI} = (V_{\interI},E_{\interI})$
of an interpretation $\interI$ is a simple undirected graph whose nodes are
domain elements from $\DeltaI$ and an edge exists between two nodes when there
is a role that connects them in $\interI$.

Let $
    \sigmaO_c \triangleq \sigma_c \cup \{ \conceptC_{\indvo, \roler} \mid
    \indvo \in \sigma_i, \roler \in \sigma_r \}
$
and $\sigmaO_r \triangleq \sigma_r \cup \{ \roler_{\indvo} \mid \indvo \in \sigma_i \}$.
By the
\emph{nominal-enrichment} of a $\vocabV \triangleq (\sigma_i, \sigma_c, \sigma_r)$-interpretation
$\interI$ we mean the $\vocabV^{\extO} \triangleq(\sigma_i, \sigmaO_c, \sigmaO_r)$-interpretation
$\interI_{\extO}$ defined in the following steps. We encourage the reader to
consult the example depicted below while going through the steps:

\begin{itemize}\itemsep0em
    \item \textbf{(A)}
        First, we get rid of unreachable elements from $\interI$. More
        precisely, let $\interJ$ to be the substructure of $\interI$ restricted
        to the set of all elements reachable in (finitely-many steps) from
        $\elD$ in $\mathsf{G}_\interI$. Without the loss of generality, we can
        assume that all $\indvo^{\interI}$ for $\indvo \in \sigma_i$ are reachable.
    \item \redText{\textbf{(B)}}
        For each pair $(\elD, \indvo) \in \DeltaI \times \sigma_i$ such that
        there is a $\roler$-connection from $\elD$ to
        $\indvo^{\interI}$, we insert a ``trampoline'' element
        labelled by the unique concept name $\redText{\conceptC_{\indvo, \roler}}$ and we
        \redText{$\roler$-connect} it with $\elD$.

        Trampoline elements are used to bookkeep information about connections
        between elements and named elements. Let $\interJ$ be the resulting
        interpretation.
    \item \textbf{(C)}
        We next divide $\interJ$ into components. Let $\interJ_\indvo$ for
        $\indvo \in \sigma_i \cup \{ \elD \}$ (with $\elD$ being the root element)
        be induced subinterpretations of $\interJ$ obtained by removing all elements
        $\{ \indvo^{\interI} \mid \indvo \in \sigma_i \}$ from $\interJ$ except
        the element mentioned in the subscript (that serve the role of
        distinguished elements of the components). In each component $\interJ_\indvo$,
        we take only elements reachable from $\indvo$. Take $\interJ'$ to be the
        disjoint sum of the components.

    \item \orangeText{\textbf{(D)}}
        In the last step, we will link components. For all $\indvo \in \sigma_i$, take
        $\textrm{dist}_\indvo$ to be the length of the shortest path from
        $\elD$ to $\indvo^{\interI}$ in $\mathsf{G}_\interI$. We will
        connect $\elD$ to $\indvo^{\interJ'}$ by a dummy path of length
        precisely $\textrm{dist}_\indvo$. Thus, we introduce dummy elements
        $\orangeText{\elD_1^{\indvo}, \ldots, \elD_{\textrm{dist}_\indvo-1}^{\indvo}}$
        to $\Delta^{\interJ'}$ and employ the fresh role name $\orangeText{\roler_{\indvo}}$,
        whose interpretation will contain precisely the pairs $\orangeText{(\elD,
        \elD_1^{\indvo}), (\elD_1^{\indvo}, \elD_2^{\indvo}),
        \ldots, (\elD_{\textrm{dist}_\indvo-1}^{\indvo},
        \indvo^{\interJ'})}$. The resulting interpretation is the desired
        $\interI_{\extO}$.
\end{itemize}

\begin{center}

\input{graphs/nominals_reduction.tex}

\end{center}

Let $\transO$ be the described transformation, mapping $(\interI, \elD)$ to $(\interI_{\extO}, \elD)$.

\begin{lemma}\label{prop:bISelfO-game-reduction}
    Let $k \in \Nomega$. Duplicator has a winning strategy in a $k$-round $\ALCO(\vocabV)$-bisimulation game
    on $(\interI, \elD)$ and $(\interJ, \elE)$ iff she has a winning strategy
    in a $k$-round $\ALC(\vocabV^{\extO})$-bisimulation game on
    $(\trans{\extO}(\interI), \elD)$ and $(\trans{\extO}(\interJ), \elE)$.
\end{lemma}
\begin{proof}[\textbf{Proof} ($\Longrightarrow$)]

    We proceed with the proof by induction on $k$, the depth parameter.
    Interpretation of concept names for distinguished elements is left
    unchanged by $\trans{\extO}$, hence Duplicator has a winning strategy
    in the $0$-round bisimulation game. Suppose now that the implication
    holds for games with at most $k$ rounds and let us show it holds for
    games with $k{+}1$ rounds. Suppose that Duplicator has a winning
    strategy in any $k{+}1$-round $\ALCO(\vocabV)$-bisimulation game. Let
    $(\trans{\extO}(\interI), s \elD ; \trans{\extO}(\interJ), s' \elE)$
    be a configuration of the $\ALC(\vocabV^{\extO})$-bisimulation game
    following the promised (by inductive hypothesis) $k$-round winning
    strategy of Duplicator. We will show how to proceed with the next step
    of the game. W.l.o.g. assume that Spoiler selected
    $\trans{\extO}(\interI)$ and decided to choose an element $\elD'$;
    we need to reply with an element $\elE'$ in the second structure.
    There are the following cases:

    \begin{enumerate}
        \item Spoiler chooses a dummy element. We reply with the
            corresponding element, which can be done without any problems since
            dummy paths of length at most $k{+}1$ leading to named elements
            have equal lengths in both interpretations. Dummy paths longer than
            $k{+}1$ are clearly equal up to $k{+}1$ elements.

        \item $\elD'$ selected by Spoiler is a trampoline. Notice that we have
            defined the trampolines in such a way that they reflect all possible
            connections to constants. Hence, by having $k{+}1$ rounds winning strategy
            in $\ALCO(\vocabV)$-bisimulation game, it implies that the elements
            reachable within $k$ steps must have had the same connections to
            constants, which means that Duplicator can respond with a trampoline
            of equal concept names.

        \item Spoiler chooses a constant $\indvo^\interI$. The only way which
            we could access a constant was via a dummy path of length at most $k$,
            which means that $d$, $e$ were on the paths labelled by the $\roler_\indvo$,
            thus they lead to the same constants, $\indvo^\interI$ and $\indvo^\interJ$,
            respectively.

        \item Spoiler chooses an ``ordinary`` element $\elD'$, that is, an
            element which does not match any of the above conditions. Then it
            means that $\elD'$ was a copy of an element in the original
            interpretation, thus, we can follow the same move that was
            made in the original interpretation by $\ALCO(\vocabV)$-winning strategy.
    \end{enumerate}

\end{proof}
\begin{proof}[($\Longleftarrow$)]

    We again proceed by induction on $k$. The base case proceeds analogously
    to the previous implication.
    Suppose now that the implication holds for games with at most $k$ rounds and
    let us show it holds for games with $k{+}1$ rounds. Suppose that
    Duplicator has a winning strategy in any $k{+}1$-round
    $\ALC(\vocabV^\extO)$-bisimulation game.
    Let $(\interI s \elD; \interJ s' \elE)$ be a configuration of the
    $\ALC\extO(\vocabV)$-bisimulation game following the promised (by
    inductive hypothesis) $k$-round winning strategy of Duplicator. We will
    show how to proceed with the next step of the game. W.l.o.g. assume that Spoiler
    selected $\interI$ and decided to choose an element
    $\elD'$; we need to reply with an element $\elE'$ in the
    second structure. There are the following cases:

    \begin{enumerate}

        \item Spoiler chooses a constant $\indvo^\interI$ via role $\roler$.
            From $\ALC(\vocabV^\extO)$-winning strategy, this means
            that in the $\trans{\extO}(\interI)$ there must have been a
            trampoline which encodes the possible connections to a constant,
            thus there was also a trampoline in $\trans{\extO}(\interJ)$
            with the same concept names, which implies that there are the same
            connections to constants from $\elD$ and $\elE$, hence,
            Duplicator can choose a constant $\indvo^\interJ$ using also $\roler$.

        \item Spoiler jumps out of the constant, \ie he was in $\indvo^\interI$
            and now using role $\roler$ selects $\elD'$ that is not a
            constant. Should $\indvo^\interI$ be accessible within $k$ steps,
            it means that we can access it in $\trans{\extO}(\interI)$ using a
            dummy path of length $ \leq k$. The outgoing connections from
            constants were restored in $\trans{\extO}(\interI)$, henceforth, from
            the constant $\indvo^{\trans{\extO}(\interI)}$ we also have a $\roler$
            connection to a copy of the element $\elD'$. This implies that
            according to $\ALC(\vocabV^\extO)$-winning strategy,
            we have a $\roler$ move to an element $\elE'$ in
            $\trans{\extO}(\interI)$. Since $\elE'$ cannot be a constant,
            it is a direct copy of an element from $\interJ$, which gives us
            a valid response for Duplicator.

        \item Spoiler chooses an ``ordinary`` element $\elD'$, that is, an
            element which does not match any of the above conditions.
            Notice that this means neither $\elD$ nor $\elD'$ can be
            a constant. That means that we have a copy of both of the elements
            $\elD$ and $\elD'$ along with all the connections between
            them, which means that Duplicator can respond following the
            $k{+}1$ step of the $\ALC(\vocabV^\extO)$-winning strategy.
    \end{enumerate}

\end{proof}

\subsection{Combining reductions}

We wrap up the above reductions, with a goal that the winning strategy of
Duplicator in a $\DL{L}\Phi$-bisimulation game is equivalent to the winning
strategy in a certain $\ALC$-bisimulation game. Note that the order of
applications of reduction matters, \eg we should apply first the $\transI$
reduction, and only then $\transb$; otherwise we will not get all possible
combinations of roles with inverses. Hence, we first proceed with $\transSelf$
reduction, then with $\transI$, with $\transb$ and finally with $\transO$. Let
$\trans{\Phi}$ be a composition of reductions for extensions $\Phi \in \extSet$
in the above order.

\begin{theorem}\label{theorem:sublogics-game-reduction}
    Let $k \in \Nomega$ and $\DL{L}\Phi$ satisfy $\ALC \subseteq \DL{L}\Phi \subseteq \ALCOIbSelf$.
    Then Duplicator has a winning strategy in a $k$-round $\DL{L}\Phi(\vocabV)$-bisimulation game
    on $(\interI, \elD)$ and $(\interJ, \elE)$ iff she has a winning strategy
    in a $k$-round $\DL{L}(\vocabV^{\Phi})$-bisimulation game on
    $(\trans{\Phi}(\interI), \elD)$ and $(\trans{\Phi}(\interJ), \elE)$.
\end{theorem}
\begin{proof}
    The key idea here is grounded on the composition of the reduction
    functions. Given $\Phi$, we simply apply consecutively Propositions
    \ref{prop:bISelf-game-reduction}--\ref{prop:b-game-reduction}
    and~\cref{prop:bISelfO-game-reduction}.
    % Based
    % on the given logic extensions $\Phi$, we apply consecutively
    % ~\cref{prop:bISelfO-game-reduction,prop:bISelf-game-reduction,prop:bI-game-reduction,prop:b-game-reduction},
    % following closely the order imposed by $\Phi_\functorOrd$ which immediately gives us the thesis.
\end{proof}

%% file: graphs/self_reduction.tex
\tikzset{every picture/.style={line width=0.75pt}} %set default line width to 0.75pt        

\begin{tikzpicture}[x=0.75pt,y=0.75pt,yscale=-1,xscale=1]
%uncomment if require: \path (0,984); %set diagram left start at 0, and has height of 984

%Straight Lines [id:da20427695230310206] 
\draw    (50,650.33) -- (129.29,600.89) ;
\draw [shift={(130.99,599.83)}, rotate = 148.06] [color={rgb, 255:red, 0; green, 0; blue, 0 }  ][line width=0.75]    (10.93,-4.9) .. controls (6.95,-2.3) and (3.31,-0.67) .. (0,0) .. controls (3.31,0.67) and (6.95,2.3) .. (10.93,4.9)   ;
%Straight Lines [id:da4887541506840023] 
\draw    (51.71,651.38) -- (68.51,661.65) -- (130.99,699.83) ;
\draw [shift={(50,650.33)}, rotate = 31.43] [color={rgb, 255:red, 0; green, 0; blue, 0 }  ][line width=0.75]    (10.93,-4.9) .. controls (6.95,-2.3) and (3.31,-0.67) .. (0,0) .. controls (3.31,0.67) and (6.95,2.3) .. (10.93,4.9)   ;
%Curve Lines [id:da1829426770581435] 
\draw    (129.37,597.33) .. controls (90.39,550.97) and (173.56,550.01) .. (135.03,595.77) ;
\draw [shift={(133.82,597.17)}, rotate = 311.29] [color={rgb, 255:red, 0; green, 0; blue, 0 }  ][line width=0.75]    (10.93,-4.9) .. controls (6.95,-2.3) and (3.31,-0.67) .. (0,0) .. controls (3.31,0.67) and (6.95,2.3) .. (10.93,4.9)   ;
%Curve Lines [id:da26966858969991114] 
\draw    (130.99,699.83) .. controls (178.61,700.08) and (210.68,680.4) .. (212.7,651.84) ;
\draw [shift={(212.79,650.08)}, rotate = 91.57] [color={rgb, 255:red, 0; green, 0; blue, 0 }  ][line width=0.75]    (10.93,-4.9) .. controls (6.95,-2.3) and (3.31,-0.67) .. (0,0) .. controls (3.31,0.67) and (6.95,2.3) .. (10.93,4.9)   ;
%Curve Lines [id:da6433709633703621] 
\draw    (212.79,650.08) .. controls (164.37,650.08) and (132.29,669.29) .. (131.03,698.06) ;
\draw [shift={(130.99,699.83)}, rotate = 270] [color={rgb, 255:red, 0; green, 0; blue, 0 }  ][line width=0.75]    (10.93,-4.9) .. controls (6.95,-2.3) and (3.31,-0.67) .. (0,0) .. controls (3.31,0.67) and (6.95,2.3) .. (10.93,4.9)   ;
%Straight Lines [id:da17778912647178857] 
\draw    (350.34,649.83) -- (428.18,600.41) ;
\draw [shift={(429.87,599.33)}, rotate = 147.59] [color={rgb, 255:red, 0; green, 0; blue, 0 }  ][line width=0.75]    (10.93,-4.9) .. controls (6.95,-2.3) and (3.31,-0.67) .. (0,0) .. controls (3.31,0.67) and (6.95,2.3) .. (10.93,4.9)   ;
%Straight Lines [id:da8246748737542051] 
\draw    (352.04,650.89) -- (368.52,661.15) -- (429.87,699.33) ;
\draw [shift={(350.34,649.83)}, rotate = 31.9] [color={rgb, 255:red, 0; green, 0; blue, 0 }  ][line width=0.75]    (10.93,-4.9) .. controls (6.95,-2.3) and (3.31,-0.67) .. (0,0) .. controls (3.31,0.67) and (6.95,2.3) .. (10.93,4.9)   ;
%Curve Lines [id:da9180681072774521] 
\draw    (429.87,699.33) .. controls (476.64,699.58) and (508.13,679.9) .. (510.12,651.34) ;
\draw [shift={(510.2,649.58)}, rotate = 91.54] [color={rgb, 255:red, 0; green, 0; blue, 0 }  ][line width=0.75]    (10.93,-4.9) .. controls (6.95,-2.3) and (3.31,-0.67) .. (0,0) .. controls (3.31,0.67) and (6.95,2.3) .. (10.93,4.9)   ;
%Curve Lines [id:da9578999718759786] 
\draw    (510.2,649.58) .. controls (462.66,649.58) and (431.15,668.79) .. (429.91,697.56) ;
\draw [shift={(429.87,699.33)}, rotate = 270] [color={rgb, 255:red, 0; green, 0; blue, 0 }  ][line width=0.75]    (10.93,-4.9) .. controls (6.95,-2.3) and (3.31,-0.67) .. (0,0) .. controls (3.31,0.67) and (6.95,2.3) .. (10.93,4.9)   ;
%Straight Lines [id:da9021547904303717] 
\draw [color={rgb, 255:red, 0; green, 0; blue, 0 }  ,draw opacity=1 ][line width=0.75]    (261.5,650.33) -- (318.5,650.11) ;
\draw [shift={(320.5,650.1)}, rotate = 179.77] [color={rgb, 255:red, 0; green, 0; blue, 0 }  ,draw opacity=1 ][line width=0.75]    (10.93,-4.9) .. controls (6.95,-2.3) and (3.31,-0.67) .. (0,0) .. controls (3.31,0.67) and (6.95,2.3) .. (10.93,4.9)   ;
\draw [shift={(261.5,650.33)}, rotate = 179.77] [color={rgb, 255:red, 0; green, 0; blue, 0 }  ,draw opacity=1 ][line width=0.75]    (0,5.59) -- (0,-5.59)   ;
%Curve Lines [id:da413869122981003] 
\draw    (210.37,648.33) .. controls (171.39,601.97) and (254.56,601.01) .. (216.03,646.77) ;
\draw [shift={(214.82,648.17)}, rotate = 311.29] [color={rgb, 255:red, 0; green, 0; blue, 0 }  ][line width=0.75]    (10.93,-4.9) .. controls (6.95,-2.3) and (3.31,-0.67) .. (0,0) .. controls (3.31,0.67) and (6.95,2.3) .. (10.93,4.9)   ;
%Curve Lines [id:da928482148076526] 
\draw    (427.37,597.83) .. controls (388.39,551.47) and (471.56,550.51) .. (433.03,596.27) ;
\draw [shift={(431.82,597.67)}, rotate = 311.29] [color={rgb, 255:red, 0; green, 0; blue, 0 }  ][line width=0.75]    (10.93,-4.9) .. controls (6.95,-2.3) and (3.31,-0.67) .. (0,0) .. controls (3.31,0.67) and (6.95,2.3) .. (10.93,4.9)   ;
%Curve Lines [id:da6917481366574791] 
\draw    (507.87,648.33) .. controls (468.89,601.97) and (552.06,601.01) .. (513.53,646.77) ;
\draw [shift={(512.32,648.17)}, rotate = 311.29] [color={rgb, 255:red, 0; green, 0; blue, 0 }  ][line width=0.75]    (10.93,-4.9) .. controls (6.95,-2.3) and (3.31,-0.67) .. (0,0) .. controls (3.31,0.67) and (6.95,2.3) .. (10.93,4.9)   ;

\draw (127.39,544) node [anchor=north west][inner sep=0.75pt]    {$\roler$};
% Text Node
\draw (207.96,595) node [anchor=north west][inner sep=0.75pt]    {$\roles$};
% Text Node
\draw (424.68,543.5) node [anchor=north west][inner sep=0.75pt]    {$\roler$};
% Text Node
\draw (505.77,595) node [anchor=north west][inner sep=0.75pt]    {$\roles$};
% Text Node
\draw (373.51,579.63) node [anchor=north west][inner sep=0.75pt]  [color={rgb, 255:red, 208; green, 2; blue, 27 }  ,opacity=1 ]  {$\conceptC_{\Self.\roler}$};
% Text Node
\draw (519.77,647.13) node [anchor=north west][inner sep=0.75pt]  [color={rgb, 255:red, 208; green, 2; blue, 27 }  ,opacity=1 ]  {$\conceptC_{\Self.\roles}$};
% Text Node
\draw (112.3,717.83) node [anchor=north west][inner sep=0.75pt]    {$(\interI, \elD)$};
% Text Node
\draw (407.83,718.33) node [anchor=north west][inner sep=0.75pt]    {$(\interI_{\Self}, \elD)$};

\end{tikzpicture}

%% file: graphs/inverses_reduction.tex
\tikzset{every picture/.style={line width=0.75pt}} %set default line width to 0.75pt        

\begin{tikzpicture}[x=0.75pt,y=0.75pt,yscale=-1,xscale=1]
%uncomment if require: \path (0,1574); %set diagram left start at 0, and has height of 1574

%Straight Lines [id:da42902162648479014] 
\draw    (50.5,949.83) -- (129.79,900.39) ;
\draw [shift={(131.49,899.33)}, rotate = 148.06] [color={rgb, 255:red, 0; green, 0; blue, 0 }  ][line width=0.75]    (10.93,-4.9) .. controls (6.95,-2.3) and (3.31,-0.67) .. (0,0) .. controls (3.31,0.67) and (6.95,2.3) .. (10.93,4.9)   ;
%Straight Lines [id:da8955287971732084] 
\draw    (52.21,950.88) -- (69.01,961.15) -- (131.49,999.33) ;
\draw [shift={(50.5,949.83)}, rotate = 31.43] [color={rgb, 255:red, 0; green, 0; blue, 0 }  ][line width=0.75]    (10.93,-4.9) .. controls (6.95,-2.3) and (3.31,-0.67) .. (0,0) .. controls (3.31,0.67) and (6.95,2.3) .. (10.93,4.9)   ;
%Curve Lines [id:da22229727380364372] 
\draw    (129.87,896.83) .. controls (90.89,850.47) and (174.06,849.51) .. (135.53,895.27) ;
\draw [shift={(134.32,896.67)}, rotate = 311.29] [color={rgb, 255:red, 0; green, 0; blue, 0 }  ][line width=0.75]    (10.93,-4.9) .. controls (6.95,-2.3) and (3.31,-0.67) .. (0,0) .. controls (3.31,0.67) and (6.95,2.3) .. (10.93,4.9)   ;
%Curve Lines [id:da786610255465785] 
\draw    (131.49,999.33) .. controls (179.11,999.58) and (211.18,979.9) .. (213.2,951.34) ;
\draw [shift={(213.29,949.58)}, rotate = 91.57] [color={rgb, 255:red, 0; green, 0; blue, 0 }  ][line width=0.75]    (10.93,-4.9) .. controls (6.95,-2.3) and (3.31,-0.67) .. (0,0) .. controls (3.31,0.67) and (6.95,2.3) .. (10.93,4.9)   ;
%Curve Lines [id:da875361808402868] 
\draw    (213.29,949.58) .. controls (164.87,949.58) and (132.79,968.79) .. (131.53,997.56) ;
\draw [shift={(131.49,999.33)}, rotate = 270] [color={rgb, 255:red, 0; green, 0; blue, 0 }  ][line width=0.75]    (10.93,-4.9) .. controls (6.95,-2.3) and (3.31,-0.67) .. (0,0) .. controls (3.31,0.67) and (6.95,2.3) .. (10.93,4.9)   ;
%Straight Lines [id:da23725179759370763] 
\draw [color={rgb, 255:red, 0; green, 0; blue, 0 }  ,draw opacity=1 ][line width=0.75]    (262,949.83) -- (319,949.61) ;
\draw [shift={(321,949.6)}, rotate = 179.77] [color={rgb, 255:red, 0; green, 0; blue, 0 }  ,draw opacity=1 ][line width=0.75]    (10.93,-4.9) .. controls (6.95,-2.3) and (3.31,-0.67) .. (0,0) .. controls (3.31,0.67) and (6.95,2.3) .. (10.93,4.9)   ;
\draw [shift={(262,949.83)}, rotate = 179.77] [color={rgb, 255:red, 0; green, 0; blue, 0 }  ,draw opacity=1 ][line width=0.75]    (0,5.59) -- (0,-5.59)   ;
%Curve Lines [id:da7200272397104786] 
\draw    (210.87,947.83) .. controls (171.89,901.47) and (255.06,900.51) .. (216.53,946.27) ;
\draw [shift={(215.32,947.67)}, rotate = 311.29] [color={rgb, 255:red, 0; green, 0; blue, 0 }  ][line width=0.75]    (10.93,-4.9) .. controls (6.95,-2.3) and (3.31,-0.67) .. (0,0) .. controls (3.31,0.67) and (6.95,2.3) .. (10.93,4.9)   ;
%Curve Lines [id:da2708166157151397] 
\draw    (429.87,896.33) .. controls (390.89,849.97) and (474.06,849.01) .. (435.53,894.77) ;
\draw [shift={(434.32,896.17)}, rotate = 311.29] [color={rgb, 255:red, 0; green, 0; blue, 0 }  ][line width=0.75]    (10.93,-4.9) .. controls (6.95,-2.3) and (3.31,-0.67) .. (0,0) .. controls (3.31,0.67) and (6.95,2.3) .. (10.93,4.9)   ;
%Curve Lines [id:da2651825943315351] 
\draw    (431.49,998.83) .. controls (479.11,999.08) and (511.18,979.4) .. (513.2,950.84) ;
\draw [shift={(513.29,949.08)}, rotate = 91.57] [color={rgb, 255:red, 0; green, 0; blue, 0 }  ][line width=0.75]    (10.93,-4.9) .. controls (6.95,-2.3) and (3.31,-0.67) .. (0,0) .. controls (3.31,0.67) and (6.95,2.3) .. (10.93,4.9)   ;
%Curve Lines [id:da2476713926595997] 
\draw    (513.29,949.08) .. controls (464.87,949.08) and (432.79,968.29) .. (431.53,997.06) ;
\draw [shift={(431.49,998.83)}, rotate = 270] [color={rgb, 255:red, 0; green, 0; blue, 0 }  ][line width=0.75]    (10.93,-4.9) .. controls (6.95,-2.3) and (3.31,-0.67) .. (0,0) .. controls (3.31,0.67) and (6.95,2.3) .. (10.93,4.9)   ;
%Curve Lines [id:da3129559963012538] 
\draw    (510.87,947.33) .. controls (471.89,900.97) and (555.06,900.01) .. (516.53,945.77) ;
\draw [shift={(515.32,947.17)}, rotate = 311.29] [color={rgb, 255:red, 0; green, 0; blue, 0 }  ][line width=0.75]    (10.93,-4.9) .. controls (6.95,-2.3) and (3.31,-0.67) .. (0,0) .. controls (3.31,0.67) and (6.95,2.3) .. (10.93,4.9)   ;
%Curve Lines [id:da01842758458901672] 
\draw    (353.17,950.33) .. controls (400.55,949.97) and (431.99,929.64) .. (432.79,900.58) ;
\draw [shift={(350.99,950.33)}, rotate = 0.29] [color={rgb, 255:red, 0; green, 0; blue, 0 }  ][line width=0.75]    (10.93,-4.9) .. controls (6.95,-2.3) and (3.31,-0.67) .. (0,0) .. controls (3.31,0.67) and (6.95,2.3) .. (10.93,4.9)   ;
%Curve Lines [id:da29867501948748076] 
\draw    (430.58,900.6) .. controls (382.42,901.19) and (350.99,921.03) .. (350.99,950.33) ;
\draw [shift={(432.79,900.58)}, rotate = 180] [color={rgb, 255:red, 0; green, 0; blue, 0 }  ][line width=0.75]    (10.93,-4.9) .. controls (6.95,-2.3) and (3.31,-0.67) .. (0,0) .. controls (3.31,0.67) and (6.95,2.3) .. (10.93,4.9)   ;
%Curve Lines [id:da5613284053036307] 
\draw    (350.99,950.33) .. controls (351.2,979.01) and (397.78,999.76) .. (429.56,998.91) ;
\draw [shift={(431.49,998.83)}, rotate = 176.82] [color={rgb, 255:red, 0; green, 0; blue, 0 }  ][line width=0.75]    (10.93,-4.9) .. controls (6.95,-2.3) and (3.31,-0.67) .. (0,0) .. controls (3.31,0.67) and (6.95,2.3) .. (10.93,4.9)   ;
%Curve Lines [id:da17612649763673294] 
\draw    (350.99,950.33) .. controls (379.62,950.59) and (428.69,969.33) .. (431.38,997.12) ;
\draw [shift={(431.49,998.83)}, rotate = 268.43] [color={rgb, 255:red, 0; green, 0; blue, 0 }  ][line width=0.75]    (10.93,-4.9) .. controls (6.95,-2.3) and (3.31,-0.67) .. (0,0) .. controls (3.31,0.67) and (6.95,2.3) .. (10.93,4.9)   ;

% Text Node
\draw (112.8,1017.83) node [anchor=north west][inner sep=0.75pt]    {$(\mathcal{I} ,\mathrm{d})$};
% Text Node
\draw (74.45,913.05) node [anchor=north west][inner sep=0.75pt]  [rotate=-326.43]  {$\mathit{r}_{1}$};
% Text Node
\draw (125,845.3) node [anchor=north west][inner sep=0.75pt]    {$\mathit{r}_{2}$};
% Text Node
\draw (207,894.8) node [anchor=north west][inner sep=0.75pt]    {$\mathit{r}_{3}$};
% Text Node
\draw (145.8,947.6) node [anchor=north west][inner sep=0.75pt]  [rotate=-330.23]  {$\mathit{r}_{4}$};
% Text Node
\draw (188.69,991.43) node [anchor=north west][inner sep=0.75pt]  [rotate=-323.02]  {$\mathit{r}_{5}$};
% Text Node
\draw (82.59,973.61) node [anchor=north west][inner sep=0.75pt]  [rotate=-31.21]  {$\mathit{r}_{6}$};
% Text Node
\draw (404.8,1018.33) node [anchor=north west][inner sep=0.75pt]    {$(\mathcal{I}_{\mathcal{I}} ,\mathrm{d})$};
% Text Node
\draw (358.21,899.45) node [anchor=north west][inner sep=0.75pt]  [rotate=-329.04]  {$\mathit{r}_{1}$};
% Text Node
\draw (414.5,840.8) node [anchor=north west][inner sep=0.75pt]    {$\mathit{r}_{2} ,\textcolor[rgb]{0.82,0.01,0.11}{r_{2}^{-}}$};
% Text Node
\draw (407.51,947.9) node [anchor=north west][inner sep=0.75pt]  [rotate=-25.19]  {$\mathit{r}_{6}$};
% Text Node
\draw (497,891.8) node [anchor=north west][inner sep=0.75pt]    {$\mathit{r}_{3} ,\textcolor[rgb]{0.82,0.01,0.11}{r}\textcolor[rgb]{0.82,0.01,0.11}{_{3}^{-}}$};
% Text Node
\draw (438.87,945.7) node [anchor=north west][inner sep=0.75pt]  [rotate=-336.02]  {$\mathit{r}_{4} ,\textcolor[rgb]{0.82,0.01,0.11}{r}\textcolor[rgb]{0.82,0.01,0.11}{_{5}^{-}}$};
% Text Node
\draw (470.22,992.46) node [anchor=north west][inner sep=0.75pt]  [rotate=-333.16]  {$\mathit{r}_{5} ,\textcolor[rgb]{0.82,0.01,0.11}{r}\textcolor[rgb]{0.82,0.01,0.11}{_{4}^{-}}$};
% Text Node
\draw (368.11,977.87) node [anchor=north west][inner sep=0.75pt]  [rotate=-29.84]  {$\textcolor[rgb]{0.82,0.01,0.11}{r}\textcolor[rgb]{0.82,0.01,0.11}{_{6}^{-}}$};
% Text Node
\draw (390.46,920.33) node [anchor=north west][inner sep=0.75pt]  [rotate=-328.1]  {$\textcolor[rgb]{0.82,0.01,0.11}{r}\textcolor[rgb]{0.82,0.01,0.11}{_{1}^{-}}$};

\end{tikzpicture}

%% file: graphs/bool_comb_reduction.tex
\tikzset{every picture/.style={line width=0.75pt}} %set default line width to 0.75pt        

\begin{tikzpicture}[x=0.75pt,y=0.75pt,yscale=-1,xscale=1]
%uncomment if require: \path (0,1366); %set diagram left start at 0, and has height of 1366

%Curve Lines [id:da36087253064791214] 
\draw    (180.09,800.73) .. controls (131.93,801.31) and (100.5,821.16) .. (100.5,850.46) ;
\draw [shift={(182.3,800.71)}, rotate = 180] [color={rgb, 255:red, 0; green, 0; blue, 0 }  ][line width=0.75]    (10.93,-4.9) .. controls (6.95,-2.3) and (3.31,-0.67) .. (0,0) .. controls (3.31,0.67) and (6.95,2.3) .. (10.93,4.9)   ;
%Straight Lines [id:da5393501879355196] 
\draw    (100.5,850.46) -- (179.79,801.02) ;
\draw [shift={(181.49,799.96)}, rotate = 148.06] [color={rgb, 255:red, 0; green, 0; blue, 0 }  ][line width=0.75]    (10.93,-4.9) .. controls (6.95,-2.3) and (3.31,-0.67) .. (0,0) .. controls (3.31,0.67) and (6.95,2.3) .. (10.93,4.9)   ;
%Straight Lines [id:da41275305041021904] 
\draw    (102.21,851.5) -- (119.01,861.78) -- (181.49,899.96) ;
\draw [shift={(100.5,850.46)}, rotate = 31.43] [color={rgb, 255:red, 0; green, 0; blue, 0 }  ][line width=0.75]    (10.93,-4.9) .. controls (6.95,-2.3) and (3.31,-0.67) .. (0,0) .. controls (3.31,0.67) and (6.95,2.3) .. (10.93,4.9)   ;
%Curve Lines [id:da6171678892682357] 
\draw    (179.87,797.46) .. controls (140.89,751.1) and (224.06,750.14) .. (185.53,795.89) ;
\draw [shift={(184.32,797.3)}, rotate = 311.29] [color={rgb, 255:red, 0; green, 0; blue, 0 }  ][line width=0.75]    (10.93,-4.9) .. controls (6.95,-2.3) and (3.31,-0.67) .. (0,0) .. controls (3.31,0.67) and (6.95,2.3) .. (10.93,4.9)   ;
%Curve Lines [id:da6893115360914763] 
\draw    (181.49,899.96) .. controls (229.11,900.21) and (261.18,880.52) .. (263.2,851.97) ;
\draw [shift={(263.29,850.21)}, rotate = 91.57] [color={rgb, 255:red, 0; green, 0; blue, 0 }  ][line width=0.75]    (10.93,-4.9) .. controls (6.95,-2.3) and (3.31,-0.67) .. (0,0) .. controls (3.31,0.67) and (6.95,2.3) .. (10.93,4.9)   ;
%Curve Lines [id:da6126175752662701] 
\draw    (263.29,850.21) .. controls (214.87,850.21) and (182.79,869.42) .. (181.53,898.19) ;
\draw [shift={(181.49,899.96)}, rotate = 270] [color={rgb, 255:red, 0; green, 0; blue, 0 }  ][line width=0.75]    (10.93,-4.9) .. controls (6.95,-2.3) and (3.31,-0.67) .. (0,0) .. controls (3.31,0.67) and (6.95,2.3) .. (10.93,4.9)   ;
%Straight Lines [id:da15117037765086772] 
\draw [color={rgb, 255:red, 0; green, 0; blue, 0 }  ,draw opacity=1 ][line width=0.75]    (312,850.46) -- (369,850.24) ;
\draw [shift={(371,850.23)}, rotate = 179.77] [color={rgb, 255:red, 0; green, 0; blue, 0 }  ,draw opacity=1 ][line width=0.75]    (10.93,-4.9) .. controls (6.95,-2.3) and (3.31,-0.67) .. (0,0) .. controls (3.31,0.67) and (6.95,2.3) .. (10.93,4.9)   ;
\draw [shift={(312,850.46)}, rotate = 179.77] [color={rgb, 255:red, 0; green, 0; blue, 0 }  ,draw opacity=1 ][line width=0.75]    (0,5.59) -- (0,-5.59)   ;
%Curve Lines [id:da06946256095510073] 
\draw    (100.5,850.46) .. controls (148.12,850.71) and (180.19,831.02) .. (182.21,802.47) ;
\draw [shift={(182.3,800.71)}, rotate = 91.57] [color={rgb, 255:red, 0; green, 0; blue, 0 }  ][line width=0.75]    (10.93,-4.9) .. controls (6.95,-2.3) and (3.31,-0.67) .. (0,0) .. controls (3.31,0.67) and (6.95,2.3) .. (10.93,4.9)   ;
%Straight Lines [id:da31271044433479345] 
\draw    (400.5,850.03) -- (479.79,800.59) ;
\draw [shift={(481.49,799.53)}, rotate = 148.06] [color={rgb, 255:red, 0; green, 0; blue, 0 }  ][line width=0.75]    (10.93,-4.9) .. controls (6.95,-2.3) and (3.31,-0.67) .. (0,0) .. controls (3.31,0.67) and (6.95,2.3) .. (10.93,4.9)   ;
%Straight Lines [id:da3081623873175461] 
\draw    (402.21,851.08) -- (419.01,861.35) -- (481.49,899.53) ;
\draw [shift={(400.5,850.03)}, rotate = 31.43] [color={rgb, 255:red, 0; green, 0; blue, 0 }  ][line width=0.75]    (10.93,-4.9) .. controls (6.95,-2.3) and (3.31,-0.67) .. (0,0) .. controls (3.31,0.67) and (6.95,2.3) .. (10.93,4.9)   ;
%Curve Lines [id:da033823023679126374] 
\draw    (479.87,797.03) .. controls (440.89,750.67) and (524.06,749.71) .. (485.53,795.47) ;
\draw [shift={(484.32,796.87)}, rotate = 311.29] [color={rgb, 255:red, 0; green, 0; blue, 0 }  ][line width=0.75]    (10.93,-4.9) .. controls (6.95,-2.3) and (3.31,-0.67) .. (0,0) .. controls (3.31,0.67) and (6.95,2.3) .. (10.93,4.9)   ;
%Curve Lines [id:da14552156769036007] 
\draw    (481.49,899.53) .. controls (529.11,899.78) and (561.18,880.1) .. (563.2,851.54) ;
\draw [shift={(563.29,849.78)}, rotate = 91.57] [color={rgb, 255:red, 0; green, 0; blue, 0 }  ][line width=0.75]    (10.93,-4.9) .. controls (6.95,-2.3) and (3.31,-0.67) .. (0,0) .. controls (3.31,0.67) and (6.95,2.3) .. (10.93,4.9)   ;
%Curve Lines [id:da6305204092088348] 
\draw    (563.29,849.78) .. controls (514.87,849.78) and (482.79,868.99) .. (481.53,897.76) ;
\draw [shift={(481.49,899.53)}, rotate = 270] [color={rgb, 255:red, 0; green, 0; blue, 0 }  ][line width=0.75]    (10.93,-4.9) .. controls (6.95,-2.3) and (3.31,-0.67) .. (0,0) .. controls (3.31,0.67) and (6.95,2.3) .. (10.93,4.9)   ;
%Curve Lines [id:da7196688153887449] 
\draw    (181.49,899.96) .. controls (234.46,933.86) and (281.06,910.67) .. (263.83,852) ;
\draw [shift={(263.29,850.21)}, rotate = 72.68] [color={rgb, 255:red, 0; green, 0; blue, 0 }  ][line width=0.75]    (10.93,-4.9) .. controls (6.95,-2.3) and (3.31,-0.67) .. (0,0) .. controls (3.31,0.67) and (6.95,2.3) .. (10.93,4.9)   ;
%Curve Lines [id:da37269201850415734] 
\draw    (260.87,848.63) .. controls (221.89,802.27) and (305.06,801.31) .. (266.53,847.07) ;
\draw [shift={(265.32,848.47)}, rotate = 311.29] [color={rgb, 255:red, 0; green, 0; blue, 0 }  ][line width=0.75]    (10.93,-4.9) .. controls (6.95,-2.3) and (3.31,-0.67) .. (0,0) .. controls (3.31,0.67) and (6.95,2.3) .. (10.93,4.9)   ;
%Curve Lines [id:da3215761124164398] 
\draw    (560.87,848.63) .. controls (521.89,802.27) and (605.06,801.31) .. (566.53,847.07) ;
\draw [shift={(565.32,848.47)}, rotate = 311.29] [color={rgb, 255:red, 0; green, 0; blue, 0 }  ][line width=0.75]    (10.93,-4.9) .. controls (6.95,-2.3) and (3.31,-0.67) .. (0,0) .. controls (3.31,0.67) and (6.95,2.3) .. (10.93,4.9)   ;

% Text Node
\draw (162.8,929.46) node [anchor=north west][inner sep=0.75pt]    {$(\mathcal{I} ,\mathrm{d})$};
% Text Node
\draw (109.45,800.68) node [anchor=north west][inner sep=0.75pt]  [rotate=-326.43]  {$\mathit{r}_{1}$};
% Text Node
\draw (175,745.93) node [anchor=north west][inner sep=0.75pt]    {$\mathit{r}_{2}$};
% Text Node
\draw (190.8,850.23) node [anchor=north west][inner sep=0.75pt]  [rotate=-330.23]  {$\mathit{r}_{4}$};
% Text Node
\draw (230.14,896.51) node [anchor=north west][inner sep=0.75pt]  [rotate=-328.93]  {$\mathit{r}_{5}$};
% Text Node
\draw (132.59,874.24) node [anchor=north west][inner sep=0.75pt]  [rotate=-31.21]  {$\mathit{r} 7$};
% Text Node
\draw (150.45,846.68) node [anchor=north west][inner sep=0.75pt]  [rotate=-326.43]  {$\mathit{r}_{8}$};
% Text Node
\draw (139.45,830.68) node [anchor=north west][inner sep=0.75pt]  [rotate=-326.43]  {$\mathit{r}_{9}$};
% Text Node
\draw (462.8,930.03) node [anchor=north west][inner sep=0.75pt]    {$(\mathcal{I}_\extb,\mathrm{d})$};
% Text Node
\draw (400.45,825.25) node [anchor=north west][inner sep=0.75pt]  [rotate=-326.43]  {$\textcolor[rgb]{0.82,0.01,0.11}{\{\mathit{r}_{1} ,\mathit{r}_{8} ,\mathit{r}_{9}\}}$};
% Text Node
\draw (469,742.5) node [anchor=north west][inner sep=0.75pt]    {$\textcolor[rgb]{0.82,0.01,0.11}{\{\mathit{r}_{2}\}}$};
% Text Node
\draw (522.39,902.18) node [anchor=north west][inner sep=0.75pt]  [rotate=-325.84]  {$\textcolor[rgb]{0.82,0.01,0.11}{\{\mathit{r}_{5} ,\mathit{r}_{6}\}}$};
% Text Node
\draw (428.59,871.81) node [anchor=north west][inner sep=0.75pt]  [rotate=-31.21]  {$\textcolor[rgb]{0.82,0.01,0.11}{\{\mathit{r}_{7}\}}$};
% Text Node
\draw (246.14,919.51) node [anchor=north west][inner sep=0.75pt]  [rotate=-328.93]  {$\mathit{r}_{6}$};
% Text Node
\draw (257,797.6) node [anchor=north west][inner sep=0.75pt]    {$\mathit{r}_{3}$};
% Text Node
\draw (549,794.6) node [anchor=north west][inner sep=0.75pt]    {$\textcolor[rgb]{0.82,0.01,0.11}{\{}\mathit{\textcolor[rgb]{0.82,0.01,0.11}{r}}\textcolor[rgb]{0.82,0.01,0.11}{_{3}}\textcolor[rgb]{0.82,0.01,0.11}{\}}$};
% Text Node
\draw (482.39,849.18) node [anchor=north west][inner sep=0.75pt]  [rotate=-325.84]  {$\textcolor[rgb]{0.82,0.01,0.11}{\{}\textcolor[rgb]{0.82,0.01,0.11}{\mathit{r}}\textcolor[rgb]{0.82,0.01,0.11}{_{4}\}}$};

\end{tikzpicture}

%% file: graphs/nominals_reduction.tex
\tikzset{every picture/.style={line width=0.75pt}} %set default line width to 0.75pt        

\begin{tikzpicture}[x=0.75pt,y=0.75pt,yscale=-1,xscale=1]
%uncomment if require: \path (0,2519); %set diagram left start at 0, and has height of 2519

%Straight Lines [id:da2782885087319218] 
\draw    (63.81,1633.46) -- (138.96,1587.47) ;
\draw [shift={(140.67,1586.43)}, rotate = 148.54] [color={rgb, 255:red, 0; green, 0; blue, 0 }  ][line width=0.75]    (10.93,-4.9) .. controls (6.95,-2.3) and (3.31,-0.67) .. (0,0) .. controls (3.31,0.67) and (6.95,2.3) .. (10.93,4.9)   ;
%Straight Lines [id:da5306584711975346] 
\draw    (63.81,1633.46) -- (138.95,1678.52) ;
\draw [shift={(140.67,1679.55)}, rotate = 210.95] [color={rgb, 255:red, 0; green, 0; blue, 0 }  ][line width=0.75]    (10.93,-4.9) .. controls (6.95,-2.3) and (3.31,-0.67) .. (0,0) .. controls (3.31,0.67) and (6.95,2.3) .. (10.93,4.9)   ;
%Curve Lines [id:da30140215800930226] 
\draw    (139.13,1584.1) .. controls (102.14,1540.93) and (181.06,1540.04) .. (144.5,1582.64) ;
\draw [shift={(143.36,1583.95)}, rotate = 311.83] [color={rgb, 255:red, 0; green, 0; blue, 0 }  ][line width=0.75]    (10.93,-4.9) .. controls (6.95,-2.3) and (3.31,-0.67) .. (0,0) .. controls (3.31,0.67) and (6.95,2.3) .. (10.93,4.9)   ;
%Curve Lines [id:da5344059775802295] 
\draw    (140.67,1679.55) .. controls (185.86,1679.78) and (216.29,1661.45) .. (218.21,1634.86) ;
\draw [shift={(218.29,1633.22)}, rotate = 91.6] [color={rgb, 255:red, 0; green, 0; blue, 0 }  ][line width=0.75]    (10.93,-4.9) .. controls (6.95,-2.3) and (3.31,-0.67) .. (0,0) .. controls (3.31,0.67) and (6.95,2.3) .. (10.93,4.9)   ;
%Curve Lines [id:da10884796000654928] 
\draw    (218.29,1633.22) .. controls (172.35,1633.22) and (141.9,1651.11) .. (140.71,1677.9) ;
\draw [shift={(140.67,1679.55)}, rotate = 270] [color={rgb, 255:red, 0; green, 0; blue, 0 }  ][line width=0.75]    (10.93,-4.9) .. controls (6.95,-2.3) and (3.31,-0.67) .. (0,0) .. controls (3.31,0.67) and (6.95,2.3) .. (10.93,4.9)   ;
%Curve Lines [id:da5812932289487287] 
\draw    (216,1631.75) .. controls (179.01,1588.58) and (257.93,1587.69) .. (221.37,1630.29) ;
\draw [shift={(220.23,1631.6)}, rotate = 311.83] [color={rgb, 255:red, 0; green, 0; blue, 0 }  ][line width=0.75]    (10.93,-4.9) .. controls (6.95,-2.3) and (3.31,-0.67) .. (0,0) .. controls (3.31,0.67) and (6.95,2.3) .. (10.93,4.9)   ;
%Shape: Triangle [id:dp45130705675236893] 
\draw   (140.67,1679.55) -- (166.86,1728.13) -- (114.47,1728.13) -- cycle ;

%Curve Lines [id:da3137218656677716] 
\draw [color={rgb, 255:red, 208; green, 2; blue, 27 }  ,draw opacity=1 ] [dash pattern={on 3.75pt off 1.5pt}]  (347.55,1633.46) .. controls (338.72,1669.93) and (341.28,1685.86) .. (358.35,1712.78) ;
\draw [shift={(359.42,1714.44)}, rotate = 237.16] [color={rgb, 255:red, 208; green, 2; blue, 27 }  ,draw opacity=1 ][line width=0.75]    (10.93,-4.9) .. controls (6.95,-2.3) and (3.31,-0.67) .. (0,0) .. controls (3.31,0.67) and (6.95,2.3) .. (10.93,4.9)   ;
%Straight Lines [id:da15535128409624122] 
\draw    (347.55,1633.46) -- (422.71,1587.47) ;
\draw [shift={(424.41,1586.43)}, rotate = 148.54] [color={rgb, 255:red, 0; green, 0; blue, 0 }  ][line width=0.75]    (10.93,-4.9) .. controls (6.95,-2.3) and (3.31,-0.67) .. (0,0) .. controls (3.31,0.67) and (6.95,2.3) .. (10.93,4.9)   ;
%Straight Lines [id:da7924542116442495] 
\draw    (347.55,1633.46) -- (422.7,1678.52) ;
\draw [shift={(424.41,1679.55)}, rotate = 210.95] [color={rgb, 255:red, 0; green, 0; blue, 0 }  ][line width=0.75]    (10.93,-4.9) .. controls (6.95,-2.3) and (3.31,-0.67) .. (0,0) .. controls (3.31,0.67) and (6.95,2.3) .. (10.93,4.9)   ;
%Curve Lines [id:da7602045757374722] 
\draw    (422.87,1584.1) .. controls (385.89,1540.93) and (464.81,1540.04) .. (428.25,1582.64) ;
\draw [shift={(427.1,1583.95)}, rotate = 311.83] [color={rgb, 255:red, 0; green, 0; blue, 0 }  ][line width=0.75]    (10.93,-4.9) .. controls (6.95,-2.3) and (3.31,-0.67) .. (0,0) .. controls (3.31,0.67) and (6.95,2.3) .. (10.93,4.9)   ;
%Curve Lines [id:da0355076796068885] 
\draw    (424.41,1679.55) .. controls (469.6,1679.78) and (500.03,1661.45) .. (501.96,1634.86) ;
\draw [shift={(502.04,1633.22)}, rotate = 91.6] [color={rgb, 255:red, 0; green, 0; blue, 0 }  ][line width=0.75]    (10.93,-4.9) .. controls (6.95,-2.3) and (3.31,-0.67) .. (0,0) .. controls (3.31,0.67) and (6.95,2.3) .. (10.93,4.9)   ;
%Curve Lines [id:da5644234847089191] 
\draw    (502.04,1633.22) .. controls (456.09,1633.22) and (425.65,1651.11) .. (424.45,1677.9) ;
\draw [shift={(424.41,1679.55)}, rotate = 270] [color={rgb, 255:red, 0; green, 0; blue, 0 }  ][line width=0.75]    (10.93,-4.9) .. controls (6.95,-2.3) and (3.31,-0.67) .. (0,0) .. controls (3.31,0.67) and (6.95,2.3) .. (10.93,4.9)   ;
%Curve Lines [id:da8211466424979839] 
\draw    (499.74,1631.75) .. controls (462.75,1588.58) and (541.67,1587.69) .. (505.11,1630.29) ;
\draw [shift={(503.97,1631.6)}, rotate = 311.83] [color={rgb, 255:red, 0; green, 0; blue, 0 }  ][line width=0.75]    (10.93,-4.9) .. controls (6.95,-2.3) and (3.31,-0.67) .. (0,0) .. controls (3.31,0.67) and (6.95,2.3) .. (10.93,4.9)   ;
%Shape: Triangle [id:dp5618417028388099] 
\draw   (424.41,1679.55) -- (450.61,1728.13) -- (398.22,1728.13) -- cycle ;
%Curve Lines [id:da5887894583111162] 
\draw [color={rgb, 255:red, 208; green, 2; blue, 27 }  ,draw opacity=1 ] [dash pattern={on 3.75pt off 1.5pt}]  (502.04,1633.22) .. controls (514.17,1654.48) and (517.45,1669.51) .. (519.84,1696.62) ;
\draw [shift={(519.99,1698.3)}, rotate = 265.12] [color={rgb, 255:red, 208; green, 2; blue, 27 }  ,draw opacity=1 ][line width=0.75]    (10.93,-4.9) .. controls (6.95,-2.3) and (3.31,-0.67) .. (0,0) .. controls (3.31,0.67) and (6.95,2.3) .. (10.93,4.9)   ;
%Curve Lines [id:da44003212948090975] 
\draw [color={rgb, 255:red, 208; green, 2; blue, 27 }  ,draw opacity=1 ] [dash pattern={on 3.75pt off 1.5pt}]  (424.41,1679.55) .. controls (454.65,1685.45) and (463.45,1690.04) .. (480.76,1709.51) ;
\draw [shift={(481.83,1710.72)}, rotate = 228.65] [color={rgb, 255:red, 208; green, 2; blue, 27 }  ,draw opacity=1 ][line width=0.75]    (10.93,-4.9) .. controls (6.95,-2.3) and (3.31,-0.67) .. (0,0) .. controls (3.31,0.67) and (6.95,2.3) .. (10.93,4.9)   ;
%Curve Lines [id:da4579952036402253] 
\draw [color={rgb, 255:red, 208; green, 2; blue, 27 }  ,draw opacity=1 ] [dash pattern={on 3.75pt off 1.5pt}]  (184.37,1840.01) .. controls (155.14,1864.14) and (147.52,1878.44) .. (144.9,1910.04) ;
\draw [shift={(144.74,1911.99)}, rotate = 274.29] [color={rgb, 255:red, 208; green, 2; blue, 27 }  ,draw opacity=1 ][line width=0.75]    (10.93,-4.9) .. controls (6.95,-2.3) and (3.31,-0.67) .. (0,0) .. controls (3.31,0.67) and (6.95,2.3) .. (10.93,4.9)   ;
%Straight Lines [id:da9710783724862695] 
\draw    (184.37,1840.01) -- (217.18,1910.51) ;
\draw [shift={(218.02,1912.32)}, rotate = 245.05] [color={rgb, 255:red, 0; green, 0; blue, 0 }  ][line width=0.75]    (10.93,-4.9) .. controls (6.95,-2.3) and (3.31,-0.67) .. (0,0) .. controls (3.31,0.67) and (6.95,2.3) .. (10.93,4.9)   ;
%Curve Lines [id:da6701868950451859] 
\draw    (223.13,1915.34) .. controls (273.8,1941.97) and (207.04,1969.5) .. (218.88,1918.59) ;
\draw [shift={(219.27,1917.02)}, rotate = 104.37] [color={rgb, 255:red, 0; green, 0; blue, 0 }  ][line width=0.75]    (10.93,-4.9) .. controls (6.95,-2.3) and (3.31,-0.67) .. (0,0) .. controls (3.31,0.67) and (6.95,2.3) .. (10.93,4.9)   ;
%Curve Lines [id:da26951216435146597] 
\draw    (420.03,1850.12) .. controls (383.04,1806.95) and (461.96,1806.06) .. (425.4,1848.66) ;
\draw [shift={(424.25,1849.97)}, rotate = 311.83] [color={rgb, 255:red, 0; green, 0; blue, 0 }  ][line width=0.75]    (10.93,-4.9) .. controls (6.95,-2.3) and (3.31,-0.67) .. (0,0) .. controls (3.31,0.67) and (6.95,2.3) .. (10.93,4.9)   ;
%Curve Lines [id:da40765984648894404] 
\draw [color={rgb, 255:red, 208; green, 2; blue, 27 }  ,draw opacity=1 ] [dash pattern={on 3.75pt off 1.5pt}]  (424.25,1849.97) .. controls (446.59,1877.49) and (450.34,1895.57) .. (450.14,1920.88) ;
\draw [shift={(450.12,1922.84)}, rotate = 270.87] [color={rgb, 255:red, 208; green, 2; blue, 27 }  ,draw opacity=1 ][line width=0.75]    (10.93,-4.9) .. controls (6.95,-2.3) and (3.31,-0.67) .. (0,0) .. controls (3.31,0.67) and (6.95,2.3) .. (10.93,4.9)   ;
%Curve Lines [id:da04433526470394966] 
\draw [color={rgb, 255:red, 208; green, 2; blue, 27 }  ,draw opacity=1 ] [dash pattern={on 3.75pt off 1.5pt}]  (424.25,1849.97) .. controls (407.82,1876.08) and (405.1,1897.27) .. (408.43,1920.69) ;
\draw [shift={(408.69,1922.49)}, rotate = 261.23] [color={rgb, 255:red, 208; green, 2; blue, 27 }  ,draw opacity=1 ][line width=0.75]    (10.93,-4.9) .. controls (6.95,-2.3) and (3.31,-0.67) .. (0,0) .. controls (3.31,0.67) and (6.95,2.3) .. (10.93,4.9)   ;
%Shape: Triangle [id:dp5438227803289246] 
\draw   (288.71,1847.63) -- (314.9,1896.21) -- (262.51,1896.21) -- cycle ;
%Curve Lines [id:da894529774152246] 
\draw [color={rgb, 255:red, 208; green, 2; blue, 27 }  ,draw opacity=1 ] [dash pattern={on 3.75pt off 1.5pt}]  (288.71,1847.63) .. controls (309.54,1850.78) and (329.9,1861.48) .. (344.98,1877.55) ;
\draw [shift={(346.13,1878.8)}, rotate = 227.83] [color={rgb, 255:red, 208; green, 2; blue, 27 }  ,draw opacity=1 ][line width=0.75]    (10.93,-4.9) .. controls (6.95,-2.3) and (3.31,-0.67) .. (0,0) .. controls (3.31,0.67) and (6.95,2.3) .. (10.93,4.9)   ;
%Straight Lines [id:da09633304303556289] 
\draw [color={rgb, 255:red, 0; green, 0; blue, 0 }  ,draw opacity=1 ][line width=0.75]    (266.42,1633.58) -- (303.33,1633.43) ;
\draw [shift={(305.33,1633.43)}, rotate = 179.77] [color={rgb, 255:red, 0; green, 0; blue, 0 }  ,draw opacity=1 ][line width=0.75]    (10.93,-4.9) .. controls (6.95,-2.3) and (3.31,-0.67) .. (0,0) .. controls (3.31,0.67) and (6.95,2.3) .. (10.93,4.9)   ;
\draw [shift={(266.42,1633.58)}, rotate = 179.77] [color={rgb, 255:red, 0; green, 0; blue, 0 }  ,draw opacity=1 ][line width=0.75]    (0,5.59) -- (0,-5.59)   ;
%Curve Lines [id:da12962625585166587] 
\draw [color={rgb, 255:red, 245; green, 166; blue, 35 }  ,draw opacity=1 ] [dash pattern={on 3.75pt off 1.5pt}]  (184.37,1840.01) .. controls (219.38,1827.77) and (255.82,1828.49) .. (287.27,1846.79) ;
\draw [shift={(288.71,1847.63)}, rotate = 211.03] [color={rgb, 255:red, 245; green, 166; blue, 35 }  ,draw opacity=1 ][line width=0.75]    (10.93,-3.29) .. controls (6.95,-1.4) and (3.31,-0.3) .. (0,0) .. controls (3.31,0.3) and (6.95,1.4) .. (10.93,3.29)   ;
%Curve Lines [id:da9821101788357578] 
\draw [color={rgb, 255:red, 245; green, 166; blue, 35 }  ,draw opacity=1 ] [dash pattern={on 3.75pt off 1.5pt}]  (184.37,1840.01) .. controls (214.96,1819.29) and (250.16,1814.68) .. (285.54,1820.58) ;
\draw [shift={(287.16,1820.86)}, rotate = 189.97] [color={rgb, 255:red, 245; green, 166; blue, 35 }  ,draw opacity=1 ][line width=0.75]    (10.93,-3.29) .. controls (6.95,-1.4) and (3.31,-0.3) .. (0,0) .. controls (3.31,0.3) and (6.95,1.4) .. (10.93,3.29)   ;
%Curve Lines [id:da07464035474744546] 
\draw [color={rgb, 255:red, 245; green, 166; blue, 35 }  ,draw opacity=1 ] [dash pattern={on 3.75pt off 1.5pt}]  (308.17,1822.93) .. controls (344.5,1823.85) and (373.69,1830.25) .. (418.66,1849.53) ;
\draw [shift={(420.03,1850.12)}, rotate = 203.35] [color={rgb, 255:red, 245; green, 166; blue, 35 }  ,draw opacity=1 ][line width=0.75]    (10.93,-3.29) .. controls (6.95,-1.4) and (3.31,-0.3) .. (0,0) .. controls (3.31,0.3) and (6.95,1.4) .. (10.93,3.29)   ;
%Straight Lines [id:da5186510956679791] 
\draw [color={rgb, 255:red, 0; green, 0; blue, 0 }  ,draw opacity=1 ][line width=0.75]    (424.9,1754.64) -- (425.28,1783.54) ;
\draw [shift={(425.3,1785.54)}, rotate = 269.25] [color={rgb, 255:red, 0; green, 0; blue, 0 }  ,draw opacity=1 ][line width=0.75]    (10.93,-4.9) .. controls (6.95,-2.3) and (3.31,-0.67) .. (0,0) .. controls (3.31,0.67) and (6.95,2.3) .. (10.93,4.9)   ;
\draw [shift={(424.9,1754.64)}, rotate = 269.25] [color={rgb, 255:red, 0; green, 0; blue, 0 }  ,draw opacity=1 ][line width=0.75]    (0,5.59) -- (0,-5.59)   ;
%Curve Lines [id:da15378323515440928] 
\draw [color={rgb, 255:red, 208; green, 2; blue, 27 }  ,draw opacity=1 ] [dash pattern={on 3.75pt off 1.5pt}]  (502.04,1633.22) .. controls (520.6,1649.19) and (528.73,1662.24) .. (540.07,1686.98) ;
\draw [shift={(540.77,1688.51)}, rotate = 245.53] [color={rgb, 255:red, 208; green, 2; blue, 27 }  ,draw opacity=1 ][line width=0.75]    (10.93,-4.9) .. controls (6.95,-2.3) and (3.31,-0.67) .. (0,0) .. controls (3.31,0.67) and (6.95,2.3) .. (10.93,4.9)   ;

% Text Node
\draw (175.86,1818.17) node [anchor=north west][inner sep=0.75pt]    {$\mathrm{d}$};
% Text Node
\draw (132.13,1861.42) node [anchor=north west][inner sep=0.75pt]  [color={rgb, 255:red, 208; green, 2; blue, 27 }  ,opacity=1 ]  {$\mathit{r}_{1}$};
% Text Node
\draw (124.15,1914.2) node [anchor=north west][inner sep=0.75pt]  [color={rgb, 255:red, 208; green, 2; blue, 27 }  ,opacity=1 ]  {$C_{\mathrm{o}_{r} ,r_{1}}$};
% Text Node
\draw (434.16,1845.42) node [anchor=north west][inner sep=0.75pt]    {$\mathrm{o}_{\mathit{s}}$};
% Text Node
\draw (413.95,1800.22) node [anchor=north west][inner sep=0.75pt]    {$\mathit{s}_{2}$};
% Text Node
\draw (435.46,1926.16) node [anchor=north west][inner sep=0.75pt]  [color={rgb, 255:red, 208; green, 2; blue, 27 }  ,opacity=1 ]  {$C_{\mathrm{o}_{r} ,r_{2}}$};
% Text Node
\draw (445.81,1868.05) node [anchor=north west][inner sep=0.75pt]  [color={rgb, 255:red, 208; green, 2; blue, 27 }  ,opacity=1 ]  {$\mathit{r}_{2}$};
% Text Node
\draw (390.02,1925.58) node [anchor=north west][inner sep=0.75pt]  [color={rgb, 255:red, 208; green, 2; blue, 27 }  ,opacity=1 ]  {$C_{\mathrm{o}_{s} ,s_{2}}$};
% Text Node
\draw (391.39,1867.05) node [anchor=north west][inner sep=0.75pt]  [color={rgb, 255:red, 208; green, 2; blue, 27 }  ,opacity=1 ]  {$\mathit{s}_{2}$};
% Text Node
\draw (266.83,1844.73) node [anchor=north west][inner sep=0.75pt]    {$\mathrm{o}_{\mathit{r}}$};
% Text Node
\draw (281.19,1876.83) node [anchor=north west][inner sep=0.75pt]    {$\alpha $};
% Text Node
\draw (337.44,1882.09) node [anchor=north west][inner sep=0.75pt]  [color={rgb, 255:red, 208; green, 2; blue, 27 }  ,opacity=1 ]  {$C_{\mathrm{o}_{s} ,s_{1}}$};
% Text Node
\draw (308.61,1861.68) node [anchor=north west][inner sep=0.75pt]  [color={rgb, 255:red, 208; green, 2; blue, 27 }  ,opacity=1 ]  {$\mathit{s}_{1}$};
% Text Node
\draw (289.2,1809.55) node [anchor=north west][inner sep=0.75pt]  [color={rgb, 255:red, 245; green, 166; blue, 35 }  ,opacity=1 ]  {$\mathrm{d}{_{1}^{\mathrm{o}_{s}}}$};
% Text Node
\draw (227.82,1802.48) node [anchor=north west][inner sep=0.75pt]  [color={rgb, 255:red, 245; green, 166; blue, 35 }  ,opacity=1 ]  {$r_{\mathrm{o}_{s}}$};
% Text Node
\draw (354.03,1812.72) node [anchor=north west][inner sep=0.75pt]  [color={rgb, 255:red, 245; green, 166; blue, 35 }  ,opacity=1 ]  {$r_{\mathrm{o}_{s}}$};
% Text Node
\draw (228.77,1834.14) node [anchor=north west][inner sep=0.75pt]  [color={rgb, 255:red, 245; green, 166; blue, 35 }  ,opacity=1 ]  {$r_{\mathrm{o}_{r}}$};
% Text Node
\draw (269.52,1612.75) node [anchor=north west][inner sep=0.75pt]  [color={rgb, 255:red, 208; green, 2; blue, 27 }  ,opacity=1 ] [align=left] {\textbf{(B)}};
% Text Node
\draw (432.11,1756.95) node [anchor=north west][inner sep=0.75pt]   [align=left] {\textbf{(C),\textcolor[rgb]{0.96,0.65,0.14}{(D)}}};
% Text Node
\draw (343.14,1716.81) node [anchor=north west][inner sep=0.75pt]  [color={rgb, 255:red, 208; green, 2; blue, 27 }  ,opacity=1 ]  {$C_{\mathrm{o}_{r} ,r_{1}}$};
% Text Node
\draw (332.44,1626.8) node [anchor=north west][inner sep=0.75pt]    {$\mathrm{d}$};
% Text Node
\draw (514.18,1630.57) node [anchor=north west][inner sep=0.75pt]    {$\mathrm{o}_{\mathit{s}}$};
% Text Node
\draw (400.64,1680.37) node [anchor=north west][inner sep=0.75pt]    {$\mathrm{o}_{\mathit{r}}$};
% Text Node
\draw (416.9,1708.75) node [anchor=north west][inner sep=0.75pt]    {$\alpha $};
% Text Node
\draw (371.24,1659.14) node [anchor=north west][inner sep=0.75pt]    {$\mathit{r}_{1}$};
% Text Node
\draw (477.53,1670.31) node [anchor=north west][inner sep=0.75pt]    {$\mathit{s}_{1}$};
% Text Node
\draw (439.57,1624.68) node [anchor=north west][inner sep=0.75pt]    {$\mathit{r}_{2}$};
% Text Node
\draw (493.66,1581.85) node [anchor=north west][inner sep=0.75pt]    {$\mathit{s}_{2}$};
% Text Node
\draw (540.77,1688.51) node [anchor=north west][inner sep=0.75pt]  [color={rgb, 255:red, 208; green, 2; blue, 27 }  ,opacity=1 ]  {$C_{\mathrm{o}_{r} ,r_{2}}$};
% Text Node
\draw (525.57,1645.21) node [anchor=north west][inner sep=0.75pt]  [color={rgb, 255:red, 208; green, 2; blue, 27 }  ,opacity=1 ]  {$\mathit{r}_{2}$};
% Text Node
\draw (506.7,1703.28) node [anchor=north west][inner sep=0.75pt]  [color={rgb, 255:red, 208; green, 2; blue, 27 }  ,opacity=1 ]  {$C_{\mathrm{o}_{s} ,s_{2}}$};
% Text Node
\draw (473.15,1714.01) node [anchor=north west][inner sep=0.75pt]  [color={rgb, 255:red, 208; green, 2; blue, 27 }  ,opacity=1 ]  {$C_{\mathrm{o}_{s} ,s_{1}}$};
% Text Node
\draw (444.31,1693.59) node [anchor=north west][inner sep=0.75pt]  [color={rgb, 255:red, 208; green, 2; blue, 27 }  ,opacity=1 ]  {$\mathit{s}_{1}$};
% Text Node
\draw (495.57,1662.82) node [anchor=north west][inner sep=0.75pt]  [color={rgb, 255:red, 208; green, 2; blue, 27 }  ,opacity=1 ]  {$\mathit{s}_{2}$};
% Text Node
\draw (326.64,1671.25) node [anchor=north west][inner sep=0.75pt]  [color={rgb, 255:red, 208; green, 2; blue, 27 }  ,opacity=1 ]  {$\mathit{r}_{1}$};
% Text Node
\draw (209.92,1581.85) node [anchor=north west][inner sep=0.75pt]    {$\mathit{s}_{2}$};
% Text Node
\draw (155.83,1624.68) node [anchor=north west][inner sep=0.75pt]    {$\mathit{r}_{2}$};
% Text Node
\draw (193.78,1670.31) node [anchor=north west][inner sep=0.75pt]    {$\mathit{s}_{1}$};
% Text Node
\draw (87.5,1659.14) node [anchor=north west][inner sep=0.75pt]    {$\mathit{r}_{1}$};
% Text Node
\draw (133.15,1708.75) node [anchor=north west][inner sep=0.75pt]    {$\alpha $};
% Text Node
\draw (116.89,1680.37) node [anchor=north west][inner sep=0.75pt]    {$\mathrm{o}_{\mathit{r}}$};
% Text Node
\draw (224.43,1633.57) node [anchor=north west][inner sep=0.75pt]    {$\mathrm{o}_{\mathit{s}}$};
% Text Node
\draw (49.64,1624.01) node [anchor=north west][inner sep=0.75pt]    {$\mathrm{d}$};

\end{tikzpicture}

%% file: 5-game-comonads.tex
\section{Game Comonads}\label{section:game-comonads}

Having defined a family of game reductions, we are going to start employing
basic category theory primitives to define denotational semantics for
bisimulation games. In this chapter, we focus on vanilla $\ALC$.
Since $\ALC$ is a notational variant of the multi-modal logic, it
suffices to translate the work done in~\cite{relating-struct-and-pow} to the
description logic setting. Subsequently, we prove that such a definition of a
generalised game coincides with our definition of
$\ALC(\vocabV)$-bisimulation game defined in~\cref{section:bisim-games}.
This chapter may be a bit heavy for readers not familiar enough
with category~theory.

\textbf{The setting}. In what follows, we shall work in the category of pointed interpretations
$\pointedCat$ over a vocabulary $\vocabV$, where objects $\structPairA $
are $\vocabV$-pointed-interpretations, and morphisms
$ h : \structPairA \rightarrow \structPairB $ are homomorphisms
between interpretations that preserve the distinguished element,
\ie $ h \: d = e $.
With $ \comonadFunctorPhi $, we will denote the corresponding game comonad,
where $k$ is the depth parameter and $ \Phi \subseteq \extSet$
parametrizes the set of language extensions. We will be a bit
careless and write $ \comonadFunctorOver{\extI\extO}$ in place of
$ \comonadFunctorOver{\{\extI, \extO\}} $, or likewise, $\comonadFunctorOver{}$
to denote $\comonadFunctorOver{\{\}}$.

\subsection{A comonad for $\ALC$}
We start with introducing the comonad for $\ALC$, which will be the base for the further ones.

\begin{definition}[$\ALC$-comonad]
    For every $k \ge 0$, we define a comonad $ \comonadFunctor$ on
    $\pointedCatALC$,\footnote{Notice $\emptyset$ in place of $\sigma_i$. This
    is because $\ALC$-concepts cannot speak about individual names.} where
    $\comonadFunctor$ unravels\footnote{For the notion of unravelling consult
    e.g.~\cite[Definition 3.21]{dlbook}.} $(\interI, \elD)$ from $ \elD$,
    up to depth $k$. More precisely:

    \begin{itemize}\itemsep0em
        \item
            The domain of $\comonadStructA $ is composed of sequences $[a_0,
            \roler_0, a_1, \roler_2, \ldots] \in \DeltaI (\sigma_r \DeltaI)^*$,
            where we additionally require that $(a_i, a_{i{+}1}) \in
            \roler_i^{\interI}$ and $a_0 = \elD$. The singleton sequence
            $[\elD]$ serves as the distinguished element of
            $\comonadStructA$.

        \item
            The functorial action on morphisms for $\comonadFunctor$ satisfies:

            \begin{center}
                $ \comonadFunctor (h : \structPairA \rightarrow \structPairB) : \comonadFunctor \structPairA \rightarrow \comonadFunctor \structPairB$ \\
                $ (\comonadFunctor \: h) [a_0, \alpha_1, a_1, ...,  \alpha_j, a_j] = [h \: a_0, \alpha_1, h \: a_1, ...,  \alpha_j, h \: a_j] $
            \end{center}
        \item
            The map $ \counitOf{\structA} : \comonadStructA \rightarrow \structPairA $
            sends a sequence to its last element.
        \item
            Concept names $\conceptC \in \sigma_c$ are interpreted such that $
            \unaryRel{\comonadStructA}{s} $ iff
            $ \unaryRel{\structA}{\counitOf{\structA}s} $.
        \item
            For role names $\roler \in \sigma_r$, we put $(s,t) \in
            \roler^{\comonadStructA}$ iff there is $d' \in \DeltaI$ so that
            $t = s[\roler, d']$.
        \item
            For a morphism $ \kleisliHomo $, we define Kleisli coextension $\kleisliExtHomo$ recursively
            by $h^* [d] = [e]$ and $h^{*} (s[\alpha, d']) = h^{*} (s)[\alpha, h(s[\alpha, d'])])$.
    \end{itemize}
\end{definition}

% TODO add note what this lemma is about

Having defined the structure, we now need to prove that it indeed forms a comonad
in the category-theoretic sense. We shall prove that $\comonadFunctor$ is a
functor, $\counit$ and $(\cdot)^*$ behave well and that the triple
$ (\comonadFunctor, \counit, (\cdot)^*) $ fulfils the comonad laws.
We start with a small lemma that shall be used later in the proofs:

\begin{lemma}\label{lemma:counit-functor-commute}
    The following diagram in $\pointedCatALC$ category commutes

    \begin{center}
        % https://tikzcd.yichuanshen.de/#N4Igdg9gJgpgziAXAbVABwnAlgFyxMJZABgBpiBdUkANwEMAbAVxiRAQF9T1Nd9CUZAIxVajFmxBce2PASJDyo+s1aIQU7iAyz+C0iOoqJ6qaJhQA5vCKgAZgCcIAWyQAmajghIAzNQZ0AEYwDAAKvHICIAwwdjggRuJqIAA6KQDGEExguACCAAScWo4uSIogXr6JqmwAFgnRQSHhuvLqDliWtfHSICWuiGQV3ogeYjXqaZnOBHRQAGLZ6V4OAPL5aYj59b39SEOViOXGyVNZOTi5wHAcDQHBYRF67Z3dZhxAA
        \begin{tikzcd}[row sep=large, column sep = large]
            \comonadStructA \arrow[r, "\comonadFunctor h"] \arrow[d, "\counitA"'] & \comonadStructB \arrow[d, "\counitB"] \\
            A \arrow[r, "h"']                                        & B
        \end{tikzcd}
    \end{center}
\end{lemma}
\begin{proof}
    Let $ s = \sequence \in \comonadStructA$. Then
    \begin{align*}
        h (\counitA \: s) &= h \: a_j       && \text{def. } \counitA \\
        &= \counitB [h \: a_j]              && \text{def. } \counitB \\
        &= \counitB [h \: a_0, \alpha_1, h \: a_1, ...,  \alpha_j, h \: a_j] && \text{def. } \counitB \\
        &= \counitB (\comonadFunctor \: h \: s)    && \text{def. }  \comonadFunctor \: h
    \end{align*}
\end{proof}

\begin{proposition}\label{prop:DL-is-a-functor}
    $\comonadFunctor$ is a functor
\end{proposition}
\begin{proof}
    We need to prove two properties

\noindent \bulletname{1}
        $\comonadFunctor$ maps objects to objects and morphisms to morphisms.

        \noindent \textbf{Objects.}
        For an interpretation $\structA$, its unravelling $\comonadStructA$ is
        also an interpretation over $(\sigmaTriple)$ which follows from the standard
        results (see e.g.~\cite[Definition 3.21]{dlbook}).

        \noindent \textbf{Morphisms.}
        Suppose $h : \structA \rightarrow \structB \in \morphismsOf{\pointedCatALC}$
        and $s, t \in \comonadStructA$.
        \begin{align*}
            \relComonad{}{s}{t} &\iff \relA{\counitA \: s}{\counitA t}        && \text{def. $\relBare{}{\comonadStructA}$ } \\
            &\: \; \Longrightarrow \relB{h(\counitA \: s)}{h \: (\counitA t)}                  && \text{h is homomorphism} \\
            &\iff \relB{\counitA (\comonadFunctor \: h \: s)}{\counitA (\comonadFunctor \: h \: t)}     && \text{\cref{lemma:counit-functor-commute}} \\
            &\iff \relComonadB{}{\comonadFunctor \: h \: s}{\comonadFunctor \: h \: t}  && \text{def. $\relBare{}{\comonadStructB}$ }
        \end{align*}

        \noindent Concept names follow similarly. \\

\noindent \bulletname{2}
        $\comonadFunctor (g \compose f) = (\comonadFunctor \: g) \compose (\comonadFunctor \: f)$
        and $\comonadFunctor \: id_\structA = id_{\comonadFunctor \structA} $
        equations are satisfied.
        \begin{align*}
            \comonadFunctor (g \compose f) s &= [(g \compose f) \: a_0, \alpha_1, (g \compose f) \: a_1, ...,  \alpha_j, (g \compose f) \: a_j]     && \text{def. $\comonadFunctor (g \compose f)$}  \\
            &= [g (f a_0), \alpha_1, g (f \: a_1), ...,  \alpha_j, g (f \: a_j)]    && \text{def. $\comp$}  \\
            &= \comonadFunctor \: g \: [f a_0, \alpha_1, f \: a_1, ...,  \alpha_j, f \: a_j]      && \text{def. $\comonadFunctor \: g$} \\
            &= \comonadFunctor \: g \: (\comonadFunctor f s)      && \text{def. $\comonadFunctor \: f $} \\
            &= (\comonadFunctor \: g) \compose (\comonadFunctor \: f) s \\ \\
            \comonadFunctor \: id_\structA \: s &=
            [id_\structA \: a_0, \alpha_1, id_\structA \: a_1, ...,  \alpha_j, id_\structA \: a_j]  && \text{def. } \comonadFunctor \: id_\structA \\
            &= \sequence = s    && \text{def. $id_\structA$} \\
            &= id_{\comonadFunctor \structA} \: s   && \text{def. $id_{\comonadFunctor \structA}$}
        \end{align*}
\end{proof}

\begin{proposition}\label{prop:e-is-a-morphism}
    $\counitA$ is a morphism in $\pointedCatALC$
\end{proposition}
\begin{proof}
    We need to show that $\counitA$ is a homomorphism and that
    it preserves the distinguished elements.
    Suppose $\relComonad{}{s}{t}$. Then $\relA{\counitA s}{\counitA t}$
    by the definition of interpretation. A distinguished element is represented
    by a singleton $[d]$ and since counit takes the last elements it
    clearly preserves them. The case for concept names is similar. \\
\end{proof}

\begin{proposition}\label{prop:e-is-natural-trans}
    $ \counit : \comonadFunctor \longrightarrow 1_{\pointedCatALC}$ is a natural transformation.
\end{proposition}
\begin{proof}

For arbitrary $\structPairA, \structPairB \in \pointedCatALC$, we need to show that

    % https://tikzcd.yichuanshen.de/#N4Igdg9gJgpgziAXAbVABwnAlgFyxMJZABgBpiBdUkANwEMAbAVxiRAEEQBfU9TXfIRRkAjFVqMWbAMLdeIDNjwEiI8uPrNWiEACE5fJYNWkx1TVJ0ARbuJhQA5vCKgAZgCcIAWyRkQOCCQ1CS02ACsQagY6ACMYBgAFfmUhEHcsBwALHAMQD29fagCkACYo2PikoxUdBhhXHPNJbTzc-J9EMv9AxABmJtCdBzbPDuDivoHLEEzIkGi4xOTjHXSsnK4KLiA
    \begin{center}
    \begin{tikzcd}[row sep=huge, column sep = huge]
        \comonadStructA \arrow[d, "\comonadFunctor h"'] \arrow[r, "\counitA"] & \structPairA \arrow[d, "h"] \\
        \comonadStructB \arrow[r, "\counitB"']                & \structPairB
    \end{tikzcd}
    \end{center}

\noindent From~\cref{prop:e-is-a-morphism} we already know that $\counitA$ and $\counitB$
are morphisms. What is left to show is that the diagram commutes:

\begin{align*}
    (h \comp \: \counitA) \sequence &= h \: a_j                                 && \text{def. } \counitA \\
    &= \counitB \: [h \: a_0, \alpha_1, h \: a_1, ...,  \alpha_j, h \: a_j]     && \text{def. } \counitB \\
    &= (\counitB \comp \comonadFunctor \: h) \: \sequence                   && \text{def. } \comonadFunctor \: h
\end{align*}

\end{proof}

\begin{proposition}\label{prop:DL-is-a-comonad}
    The triple $ (\comonadFunctor, \counit, (\cdot)^*) $ is a comonad in Kleisli
    form on $ \pointedCatALC $
\end{proposition}
\begin{proof}

    From \cref{prop:e-is-natural-trans} we have that $\counit$ is a natural
    transformation and from~\cref{prop:DL-is-a-functor} that $\comonadFunctor$
    is a functor. We need to show now that the comonadic laws are satisfied
    and that Kleisli extension behaves as expected. Precisely, we need to prove
    the following properties:

    \begin{itemize}\itemsep0em
        \nameditem{A}{$\counitA^* = id_{\comonadStructA} $}
        \nameditem{B}{$\counit \compose f^* = f $}
        \nameditem{C}{$ (g \compose f^*)^* = g^* \compose f^* $}
        \nameditem{D}{if $h$ is a morphism in $\pointedCatALC$ then $h^*$ is a morphism in $\pointedCatALC$}
    \end{itemize}
    Let $ s'' = \sequenceOf{a}{j-2}, s' = s''[\alpha_{j-1}, a_{j-1}], s = s'[\alpha_j, a_j] $.
    We will prove the comonad laws extensionally. \\

\noindent \bulletname{A} We start by showing that Kleisli extension of counit yields an identity.
    \begin{align*}
        \counitA^* s &= (\counitA^* s')[\alpha_j, \counitA s]       && \text{def. } (-)^* \\
        &= (\counitA^* s'')[\alpha_{j-1}, \counitA s', \alpha_j, \counitA s]       && \text{def. } (-)^* \\
        &= [\counitA [a_0], \alpha_1, \counitA [a_0, \alpha_1, a_1], ..., \alpha_{j-1}, \counitA s', \alpha_j, \counitA s]       && \text{apply inductively} \\
        &= [a_0, \alpha_1, a_1, ..., \alpha_{j-1}, a_{j-1}, \alpha_j, a_j] = s      && \text{def. } \counitA \\
        &= id_{\comonadStructA} s
    \end{align*}

\noindent \bulletname{B}
Let $\structPairA, \structPairB \in \pointedCatALC$ and $ f : \comonadStructA \rightarrow \structPairB $.
Then the following diagram commutes:

% https://tikzcd.yichuanshen.de/#N4Igdg9gJgpgziAXAbVABwnAlgFyxMJZABgBpiBdUkANwEMAbAVxiRBAF9T1Nd9CUZAIxVajFm07cQGbHgJEhpEdXrNWidh1EwoAc3hFQAMwBOEALZIyIHBCRKx6tsYB6AKhDUGdAEYwGAAVeeQEQUyw9AAscKRNzK0QbOyQAJlVxDRBjOOyEh2oUxHSQH38gkP42COjYjOdNAB1GmDRsBgVtDiA

 \begin{minipage}{\linewidth}
  \centering
  \begin{minipage}{0.25\linewidth}

    \begin{center}
    \begin{tikzcd}[row sep=large, column sep=large]
    {\comonadStructA} \arrow[d, "f^*"'] \arrow[rd, "f"] &    \\
    {\comonadStructB} \arrow[r, "\counitB"']            & {\structPairB}
    \end{tikzcd}
    \end{center}

  \end{minipage}
  \hspace{0.05\linewidth}
  \begin{minipage}{0.65\linewidth}

    \begin{align*}
        (\counitB \compose f^*) s &= \counitB (f^* s) \\
        &= \counitB (f^*(s')[\alpha_{j}, f \: s)])   && \text{def. } (-)^*  \\
        &= f \: s   && \text{def. } \counitB \\
    \end{align*}

    \end{minipage}
  \end{minipage}

\noindent \bulletname{C}
Let $\structPairA, \structPairB, \structPair{K}{k} \in \pointedCatALC$ and
$
    f : \comonadStructA \rightarrow \structPairB, \: \:
    g : \comonadStructB \rightarrow \structPair{K}{k}
$.
Then the following diagram commutes:

    % https://tikzcd.yichuanshen.de/#N4Igdg9gJgpgziAXAbVABwnAlgFyxMJZABgBpiBdUkANwEMAbAVxiRAEEQBfU9TXfIRRkAjFVqMWbAELdeIDNjwEiI0mOr1mrRCADC3cTCgBzeEVAAzAE4QAtkjIgcEJGona2lgHoAqENQMdABGMAwACvzKQiDWWCYAFjhyVrYOiE4uSABMmpI6IJYphWlu1FmIuSBBoRFRgmxxicl5nroAOu0waNgMKlwUXEA
    \begin{center}
    \begin{tikzcd}[row sep=large, column sep=large]
        \comonadStructA \arrow[d, "f^*"'] \arrow[rd, "(g \comp f^*)^*"] &   \\
        \comonadStructB \arrow[r, "g^*"']            & \comonadStructOver{}{K}{k}
    \end{tikzcd}
    \end{center}

    \begin{align*}
        (g \compose f^*)^* s &= (g \compose f^*)^*(s')[\alpha_{j}, (g \compose f^*) \: s)]  && \text{def. } (-)^* \\
        &= (g \compose f^*)^*(s'')[\alpha_{j-1}, (g \compose f^*) \: s', \alpha_{j}, (g \compose f^*) \: s]  && \text{def. } (-)^* \\
        &= [(g \compose f^*) [a_0], \alpha_1, (g \compose f^*) [a_0, \alpha_1, a_1], ..., \alpha_{j-1}, (g \compose f^*) \: s', \alpha_{j}, (g \compose f^*) \: s]  && \text{ind.} \\
        &= [g (f^* [a_0]), \alpha_1, g (f^* [a_0, \alpha_1, a_1]), ..., \alpha_{j-1}, g (f^* \: s'), \alpha_{j}, g (f^* \: s)] \\
        &= (1) \newline
        \intertext{since $f^* [a_0] \sqsubseteq f^* [a_0, \alpha_1, a_1] \sqsubseteq ... \sqsubseteq f^* s' \sqsubseteq f^* s$,
        we get that}
        (1) &= g^* (f^* \: s) \\
            &= (g^* \compose f^*) \: s
    \end{align*}

\noindent \bulletname{D}
        Suppose that $h$ is a morphism in $\pointedCatALC$.
        \begin{align*}
            \relComonad{}{s}{t} &\Longrightarrow \relA{h \: s}{h \: t}         && \text{h is homo.} \\
            &\Longrightarrow \relB{\counitA (h^* \: s)}{\counitA (h^* \: t)}    && \text{by \bulletname{B}} \\
            &\Longrightarrow \relComonadB{}{h^* \: s}{h^* \: t}                  && \text{def. } \relBare{}{\comonadStructB}
        \end{align*}
\end{proof}

Having the $\ALC$-comonad defined, as the next step we introduce sufficient
categorical background required to define bisimulation games in an
abstract-enough way.

\subsection{Tree-like structures, paths and embeddings.}

A \emph{covering relation}~$\prec$ for a partial order $\leq$ is a relation
satisfying $x \prec y \triangleq x \leq y \land x \ne y \land (\forall z . x
\leq z \leq y \implies z = x \lor z = y)$. This is employed to define tree-like
structures below, which  will intuitively serve as the description of
bisimulation game strategies.

\begin{definition}
    An \emph{ordered interpretation} $\ordStructA$ is a pointed
    interpretation $(\interI, \elD)$ equipped with a partial order on
    $\DeltaI$ such that $\uparrow \! (d) \triangleq \{ d' \in \DeltaI \mid d \le d' \}$
    is a tree order that satisfies the following condition
    $(\textbf{D})$ for $x, y \in \: \uparrow \! (d)$, we have $x \prec y$
    iff $(x,y)\in \roler^{\interI}$ for some $\roler \in \sigma_r$. Morphisms
    between ordered interpretations preserve the covering relation. We put
    $\pointedOrdCatK$ to be the category of ordered interpretation as objects
    with $k$ bounding the height of the underlying tree.
\end{definition}

We next define different kinds of embeddings, essential to characterize plays.

\begin{definition}
    A morphism in $\pointedOrdCatK$ is an \emph{embedding} if it is an injective
    strong homomorphism. We write $ e : \structA \embedding \structB$
    to mean that $e$ is an embedding. Now, we define a subcategory $\pathsCat$
    of $\pointedCat$ whose objects have linear tree orders, so they
    comprise a single branch. We say that $e : P \embedding \structA $ is a
    \textit{path embedding} if $P$ is a path. A morphism
    $ f : \structA \rightarrow \structB \in \morphismsOf{\pointedOrdCatK}$ is a
    \textit{pathwise embedding} if for any path embedding
    $ e : P \embedding \structA, f \circ e$ is a path embedding.
\end{definition}

Let $\sqsubseteq$ be the lexicographical order on sequences from $\DeltaI$.
From the construction of $\pointedOrdCatK$, we can extract a free functor, for
which construction is justified by the following lemma:
\begin{lemma}
    There exists a canonical functor $F_k \: \structA = (\comonadStructA, \sqsubseteq)$.
\end{lemma}
\begin{proof}
    The proof is heavy and relies on several categorical notions that are not crucial for
    the paper hence we do not introduce them here; consult~\cite[Chapters 9
    \& 10.3]{awodeybook} instead. The goal is to describe the desired functor
    in a way such that it yields the canonical, terminal resolution of a
    comonad $\comonadFunctor$. First, from~\cite[Theorem 9.5]{relating-struct-and-pow}
    we know that for any $k > 0$, the Eilenberg-Moore category $\EMcat$ is
    isomorphic to $\pointedOrdCatK$. Having that, we can observe that there is
    a forgetful functor $U_k : \pointedOrdCatK \rightarrow \pointedCat $
    mapping $\ordStructA$ to $\structPairA$ which forgets the partial order.
    Thus, we can employ the result that follows from~\cite[Theorem
    9.6]{relating-struct-and-pow} to infer that the functor $U_k$ has a right
    adjoint $F_k$. The relationship between introduced categories is depicted
    on the diagram below, where the arrow from $\pathsCat$ to $\pointedCat$ is the
    evident inclusion functor.

     % Diagram with relation of introduced categories
     \begin{center}
         % % https://tikzcd.yichuanshen.de/#N4Igdg9gJgpgziAXAbVABwnAlgFyxMJZARgBoAGAXVJADcBDAGwFcYkQAlAfQGsQBfUuky58hFACYK1Ok1btuAHUXYA5gFt6AoSAzY8BIgGZpNBizaJOXdWojbh+sUXKnZF9gFEAsgAplmjgAFgBGIcAAIvwAlA66IgbiJKTEMubyVso4MAAeOCEAZsAAMlhgAMYQjGBw-AIyMFCq8ESgBQBOEOpIriA49ohkICEwYFBIALRGvemWIACqvCA0jPQjjAAKCc5W7ViqQThxHV1IQ-09NCNjk9MrazCb24a7+4fL7hkgAGJLgm2dbqDGgXRBST5zABSXAA8sdAUgTH0Br04EEsAUjohetdxtj-iATkCACwglH8Sj8IA
         \begin{tikzcd}
            \EMcat \cong \pointedOrdCatK \arrow[r, "U_k"', bend right] & \pointedCat \arrow[l, "F_k"', bend right] & \pathsCat \arrow[l]
         \end{tikzcd}
     \end{center}

    The comonad arising from $F \dashv U$ adjunction is precisely $\comonadFunctor$.
\end{proof}

\subsection{A categorical view on games}

Given a sufficient background, we can move on to the main result, namely, the
characterisation of $\equiv^{\ALC_k}$ in the language of category theory. We
start with defining what it means for a morphism in $ f : \structA
\rightarrow \structB \in \morphismsOf{\pointedOrdCatK}$ to be \textit{open}.
This holds if, whenever we have a commutative square as on the LHS then there is an
embedding $Q \embedding \structA$ such that the diagram on the RHS commutes.

 \begin{minipage}{\linewidth}
      \centering
      \begin{minipage}{0.45\linewidth}
\begin{center}
    % https://tikzcd.yichuanshen.de/#N4Igdg9gJgpgziAXAbVABwnAlgFyxMJZABgBpiBdUkANwEMAbAVxiRAAUQBfU9TXfIRRkAjFVqMWbADrSAtnRwALOAGMATsACCXbrxAZseAkRHlx9Zq0QgAinr5HBp0mOqWpN2QuVrNAIV0ucRgoAHN4IlAAM3UIOSQyEBwIJDNkuiwGNhxM7J4YuITEJJSkACZqXKycvIcQWPiKqtTEAGYquptq-P1G4vSy9uoGOgAjGAZ2fmMhEHUsMKUcEHdJawbuCi4gA
    \begin{tikzcd}
    P \arrow[d, tail] \arrow[r, tail] & Q \arrow[d, tail] \\
    \structA \arrow[r, "f"']       & \structB
    \end{tikzcd}
\end{center}
      \end{minipage}
      \hspace{0.05\linewidth}
      \begin{minipage}{0.45\linewidth}
\begin{center}
    % https://tikzcd.yichuanshen.de/#N4Igdg9gJgpgziAXAbVABwnAlgFyxMJZABgBpiBdUkANwEMAbAVxiRAAUQBfU9TXfIRRkAjFVqMWbADrSAtnRwALOAGMATsACCXbrxAZseAkRHlx9Zq0QgAinr5HBp0mOqWpN2QuVrNAIV0ucRgoAHN4IlAAM3UIOSQyEBwIJDNkuiwGNhxM7J4YuITEJJSkACZqXKycvIcQWPiKqtTEAGYquptq-P1G4vSy9uoGOgAjGAZ2fmMhEHUsMKUcEHdJawb6-ubk1sGujJrgriA
    \begin{tikzcd}
    P \arrow[d, tail] \arrow[r, tail] & Q \arrow[d, tail] \arrow[ld, tail] \\
    \structA \arrow[r, "f"']       & \structB
    \end{tikzcd}
\end{center}
      \end{minipage}
  \end{minipage}

\noindent

Finally, we can define \textit{back-and-forth equivalence} $\backAndForth$
between objects in $\pointedCat$, intuitively corresponding to conditions (b)
and (c) from the definition of a bisimulation. This holds if there is an object
$R$ in $\pointedOrdCatK$ and a span of open pathwise embeddings such that:

\begin{center}
    % https://tikzcd.yichuanshen.de/#N4Igdg9gJgpgziAXAbVABwnAlgFyxMJZARgBoAGAXVJADcBDAGwFcYkQAlEAX1PU1z5CKcqWLU6TVuwBiAfQDWAAgA6KgLb0cACzgBjAE7AAgtx58QGbHgJEATGIkMWbRCHnK1mnfqMAhM24JGCgAc3giUAAzAwh1JFEQHAgkYl5o2PjEROSkOyDuIA
    \begin{tikzcd}
                    & R \arrow[ld] \arrow[rd] &                 \\
    F_k \structPairA &                         & F_k \structPairB
    \end{tikzcd}
\end{center}

We shall now define a back-and-forth game $\backAndForthGameOf{\Phi}$ played
between the interpretations $\structPairA$ and $\structPairB$. Positions of the
game are pairs $(s, t) \in \comonadStructProdPhi$. We define a relation
$\winningPos \subseteq \comonadStructProdPhi$ as follows. A pair $(s, t)$ is in
$\winningPos$ iff for some path $P$, \textit{path embeddings} $e_1 : P
\embedding \structA$ and $e_2 : P \embedding \structB$ , and $p \in P$, $s =
e_1 \: p$ and $t = e_2 \: p$. The intention is that $\winningPos$ picks out the
winning positions for Duplicator. At the start of each round of the game, the
position is specified by $(s, t) \in \comonadStructProdPhi$. The initial
position is $([d], [e])$. The round proceeds as follows. Spoiler either chooses
$s' \succ s$, and Duplicator must respond with $t' \succ t$, producing the new
position $(s', t')$; or Spoiler chooses $t'' \succ t$, and Duplicator must
respond with $s'' \succ s$, producing the new position $(s'', t'')$. Duplicator
wins the round if she can respond, and the new position is in
$\winningPos$. We follow the same notation convention as for
$\comonadFunctorPhi$ with respect to extensions $\Phi$ of the game
$\backAndForthGameSymbol{\Phi}$. The following theorem follows
from~\cite[Theorem 10.1]{relating-struct-and-pow}.

\begin{theorem}\label{bnf-game-theorem}
    Duplicator has a winning strategy in $\backAndForthGame$ game if and only if $\backAndForth$.
\end{theorem}

The above theorem with the aforementioned definitions were just slight variations
of theorems and notions presented in~\cite{relating-struct-and-pow}. We have
accommodated them to the description logic setting and now we will glue them
together with our definition of the bisimulation game
from~\cref{section:bisim-games}.

\begin{theorem}\label{bnf-game-is-bisimiulation-theorem}
    Given interpretations $\structPairA$ and $\structPairB$, the
    $\backAndForthGame$ game for the $\comonadFunctor$ comonad is
    equivalent to the $k$-round $\ALC(\vocabV)$-bisimulation game between $\structPairA$ and $\structPairB$.
\end{theorem}
\begin{proof}
    % TODO: more details here
    First, note that configurations and the moves are structurally the same in
    both games. Hence, by induction over $k$ it suffices to show that the
    winning conditions coincide.

    \noindent \textbf{Base}. Let $k = 0$ and suppose $([d], [e]) \in \winningPos$.
    That holds iff there are path embeddings $e_1 : P \embedding \structA$,
    $e_2 : P \embedding \structB$ and $p \in P$ such that $e_1 \: p = [d]$ and
    $e_2 \: p = [e]$. By strong homomorphism property, $d$ is in $\vocabV$-harmony
    with $p$, which in turn is in $\vocabV$-harmony with $d$, which
    by transitivity of $\vocabV$-harmony concludes this case.

    \noindent \textbf{Step}. Assume that the proposition holds for all $i \le k$.
    We need to show that the winning conditions coincide for games of length $k+1$.
    Suppose $s = s'[\alpha_s, d'], t = t'[\alpha_t, e']$ and $(s, t) \in \winningPos$.
    That holds iff there are path embeddings
    $e_1 : P \embedding \structA$, $e_2 : P \embedding \structB$ and $p \in P$
    such that $e_1 \: p = s$ and $e_2 \: p = t$. By definition of $\winningPos$
    relation, we get that $(s', t') \in \winningPos$ and hence, by the induction
    hypothesis, $s, t$ are a valid winning configuration in $\ALC$ game.
    It remains to show that $[\alpha_s, d']$ and $[\alpha_t, e']$ are valid
    moves leading to winning positions. From $e_1 \: p = s$ and $e_2 \: p = t$
    we immediately get that $\alpha_s = \alpha_t$ and since $e_1, e_2$ are
    embeddings we have that $d'$ is in $\vocabV$-harmony with $p$ which in turn is in
    $\vocabV$-harmony with $e'$, hence by transitivity of $\vocabV$-harmony, we are done.
\end{proof}

By applying \cref{bnf-game-theorem}, \cref{bnf-game-is-bisimiulation-theorem}
and \cref{fact:games-bisimulations-and-concept-eq}, we derive our first result
on comonadic semantics for description logic games, namely:

\begin{theorem}
    $ \structPairA \equiv^{\ALC_k} \structPairB \; \iff \; \backAndForth$.
\end{theorem}

%% file: 6-extensions.tex
%!TEX root = main.tex

\section{Comonads for extensions of $\ALC$}
\label{section:extensions}

We have defined description logic comonad in the previous chapter
and in~\cref{section:game-reductions} we have constructed a family of game
reductions that eliminate the logic extensions. By leveraging cautious
categorical operations, we now combine these two and arrive at having
game comonads for all considered extensions of $\ALC$.

\subsection{A generalized framework for extensions}

The approach that we undertook relies on an observation that we had based on
how $I$-morphisms were incorporated in~\cite{relating-struct-and-pow}. In our
case, relative comonads serve as a tool to start within the base category where
our objects live and then enrich the interpretations encoding the
additional capabilities available in bisimulation games for richer logics. We
do this via the already-presented reductions
from~\cref{section:game-reductions}, followed by the notion of unravelling
using $\comonadFunctor$ defined in~\cref{section:game-comonads}, all
established in a generalised framework using relative comonads.

\begin{definition}
    A \emph{vocabulary-map} $\vocabMap$ is a triple
    $(\vocabMap_i, \vocabMap_c, \vocabMap_r) : \Ilang \times \Clang \times \Rlang \rightarrow \Ilang \times \Clang \times \Rlang$
    that maps the vocabulary $(\sigmaTriple) \longmapsto (\vocabMap_i (\sigma_i), \vocabMap_c (\sigma_c), \vocabMap_r (\sigma_r))$.
\end{definition}

\begin{definition}[Reduction functor]\label{prop:reduction-functor}
    Let $\vocabMap$ be a vocabulary map and $\reduction{f}$ a game reduction.
    A \emph{$(\reduction{f}, \vocabMap)$-reduction-functor} is a
    % full and faithful~\cite[Definition 7.1]{awodeybook}
    functor
    $ J : \pointedCat \rightarrow \pointedCatOf{\vocabMap \: \vocabV}$
    acting $\structPair{I}{d} \longmapsto \structPair{\reductionI{f} \: I}{\reductionR{f} \: d}$.
\end{definition}

While~\cref{prop:reduction-functor} is stated in a general setting, we
only consider the reductions from~\cref{section:game-reductions}.
Clearly, the functors map objects to objects. When it comes to morphisms,
however, we need to handle a certain delicacy. To make reasoning simpler, let us
focus for a moment on $\ALCSelf$. Notice that interpretations that are
$\ALC$-homomorphic are not necessarily $\ALCSelf$-homomorphic, as that would
mean that self operator is expressible in bare $\ALC$, which we know is
not the case. Consecutively, that means that homomorphic interpretations are
not necessarily homomorphic after applying $\transSelf$ reduction.

To tackle this issue, we shall submerse ourselves into a particular wide subcategory,
a subcategory containing all the objects of the category of interest.

\begin{definition}
    Given $\Phi \subseteq \extSet$, a $\Phi$-subcategory of $\pointedCat$ is
    a subcategory of $\pointedCat$ with all objects from $\pointedCat$ and
    morphisms limited to $\ALC\Phi$-homomorphisms.
\end{definition}

\begin{proof}
    We need to show that the $\Phi$-subcategory of $\pointedCat$ indeed forms a
    category. First, it is easy to see that we still have identity morphisms
    on objects. Second, $\ALC\Phi$-homomorphisms are closed under composition
    which concludes the proof.
\end{proof}

From now on, when considering a set of extensions $\Phi$, we shall work in
a $\Phi$-subcategory. In this setting, the action on morphisms for reduction
functors is an identity, as the very same homomorphism will work as per
\cref{theorem:sublogics-game-reduction}. To restrain the reader from
drowning in overly verbose notation, the underlying $\Phi$-subcategory will
be taken implicitly from the context.
% The full and faithful property comes
% from the bidirectional nature of our reduction games, \ie reductions are
% reversible and as such the reduction-functor encodes a full
% subcategory.
To sum up, we obtain a family of $(\reduction{f}_\genericExt,
\vocabMap_\genericExt)$-reduction-functors, where $\genericExt \in \extSet$ are
considered logic extensions.

\begin{definition}
    Let $\vocabMap, \vocabMap'$ be a vocabulary-maps. We say that a functor
    $ F : \pointedCat \rightarrow \pointedCatOf{\vocabMap \: \vocabV} $
    is invariant over vocabulary-maps iff for any $\vocabMap'$ it
    can be lifted to
    $ F_{\vocabMap'} : \pointedCatOf{\vocabMap' \: \vocabV} \rightarrow \pointedCatOf{\vocabMap \: (\vocabMap' \: \vocabV)} $.
    We shall omit the subscript should the coercion be unambiguous.
\end{definition}

\begin{lemma}\label{prop:vocab-maps-inv-composition}
    Invariance over vocabulary maps behaves well under composition, \ie,
    the composition of functors invariant over vocabulary maps yields a functor
    invariant over vocabulary maps.
\end{lemma}
\begin{proof}
    % TODO: actually use this map
    Let $
        F : \pointedCatOf{\vocabV} \longrightarrow \pointedCatOf{\vocabMap \: \vocabV}, \:
        G : \pointedCatOf{\vocabMap \: \vocabV} \longrightarrow \pointedCatOf{\vocabMap' \: \vocabV}
    $ be functors invariant over vocabulary maps. We want to show that
    $
        (G \: \comp \: F) : \pointedCatOf{\vocabV} \longrightarrow \pointedCatOf{\vocabMap' \: \vocabV}
    $ is invariant over vocabulary maps. Let us take any vocabulary map $\vocabMap''$.
    By assumption, we can lift $F, \: G$ to $
        F_{\vocabMap''} : \pointedCatOf{\vocabMap'' \: \vocabV} \longrightarrow \pointedCatOf{(\vocabMap \comp \vocabMap'') \: \vocabV}, \:
        G_{\vocabMap''} : \pointedCatOf{(\vocabMap \comp \vocabMap'') \: \vocabV} \longrightarrow \pointedCatOf{(\vocabMap' \comp \vocabMap'') \: \vocabV}
    $. Then such composition is of the form $
        (G_{\vocabMap''} \comp F_{\vocabMap''}) : \pointedCatOf{\vocabMap'' \: \vocabV} \longrightarrow \pointedCatOf{(\vocabMap' \comp \vocabMap'') \: \vocabV}
    $ and thus $(G \comp F)$ is invariant over vocabulary maps.

 \begin{minipage}{\linewidth}
  \centering
  \begin{minipage}{0.30\linewidth}

    \begin{center}
    \begin{tikzcd}[row sep=huge, column sep=huge]
    {\pointedCatOf{\vocabV}} \arrow[d, "F"'] \arrow[rd, "G \comp F"] &    \\
    {\pointedCatOf{\vocabMap \: \vocabV}} \arrow[r, "G"']            & {\pointedCatOf{\vocabMap' \: \vocabV}}
    \end{tikzcd}
    \end{center}

  \end{minipage}
  \hspace{0.05\linewidth}
  \begin{minipage}{0.60\linewidth}

    \begin{center}
    \begin{tikzcd}[row sep=huge, column sep=huge]
        {\pointedCatOf{\vocabMap'' \: \vocabV}} \arrow[d, "F_{\vocabMap''}"'] \arrow[rd, "(G \comp F)_{\vocabMap''}"] &    \\
        {\pointedCatOf{(\vocabMap \comp \vocabMap'') \: \vocabV}} \arrow[r, "G_{\vocabMap''}"']            & {\pointedCatOf{(\vocabMap' \comp \vocabMap'') \: \vocabV}}
    \end{tikzcd}
    \end{center}

    \end{minipage}
  \end{minipage}

\end{proof}

What we want to capture by this is that such a functor acting on
$\pointedCat$ category is natural in $\vocabV$, \ie does not depend on the
contents of the concepts or roles. It is easy to see the following facts:

\begin{observation}\label{obs:comonad-vocab-maps-invariant}
    $\comonadFunctor$ is invariant over vocabulary-maps.
\end{observation}

\begin{observation}\label{obs:reduction-functors-vocab-maps-invariant}
    $(\reduction{f}_\genericExt, \vocabMap_\genericExt)$-reduction-functors
    are invariant over vocabulary-maps.
\end{observation}

To obtain richer semantics, we shall leverage the functor composition, following
the same order as defined for the game reductions in~\cref{section:game-reductions}:

% \begin{center}
%     % https://tikzcd.yichuanshen.de/#N4Igdg9gJgpgziAXAbVABwnAlgFyxMJZABgBpiBdUkANwEMAbAVxiRACUEBfU9TXfIRQBGclVqMWbThBA8+2PASIAmMdXrNWiDnACSc3iAyLBRAMzqJW6XABGc8TCgBzeEVAAzAE4QAtkhkIDiyiMLyID7+SKLBoSoRUQGIanFI5lwUXEA
%     \begin{tikzcd}
%         \pointedCatOf{\vocabV}               \arrow[r, "\Jself"] &
%         \pointedCatOf{\vocabV^{\extSelf}}               \arrow[r, "\JI"] &
%         \pointedCatOf{\vocabV^{\extSelf\extI}}       \arrow[r, "\Jb"] &
%         \pointedCatOf{\vocabV^{\extSelf\extI\extb}}  \arrow[r, "\JO"] &
%         \pointedCatOf{\vocabV^{\extSelf\extI\extb\extO}}
%     \end{tikzcd}
% \end{center}

% https://tikzcd.yichuanshen.de/#N4Igdg9gJgpgziAXAbVABwnAlgFyxMJZABgBpiBdUkANwEMAbAVxiRAEEQBfU9TXfIRQBGclVqMWbAELdeIDNjwEiAJjHV6zVohABhOXyWC1pYeK1TdAEUML+yoclHnNknSACi3cTCgBzeCJQADMAJwgAWyQyEBwIJFEJbTZhO3CoxOp4pHVkqxBVdIjoxDycxABmNxTdSuLMquyExAAWGoLWhtL2uJbehggINDUAdjIQxjgYcQY6ACMYBgAFBxNdMKx-AAscEA6PAFYfLiA
\begin{center}
    \begin{tikzcd}[row sep=huge, column sep=huge]
        \pointedCatOf{\vocabV} \arrow[r, "\Jself"] & \pointedCatOf{\vocabV^{\extSelf}} \arrow[r, "\JI"]                                    & \pointedCatOf{\vocabV^{\extSelf\extI}} \arrow[d, "\Jb"] \\
                         & \pointedCatOf{\vocabV^{\extSelf\extI\extb\extO}} \arrow["\comonadFunctor"', loop, distance=5em, in=215, out=145] & \pointedCatOf{\vocabV^{\extSelf\extI\extb}} \arrow[l, "\JO"]
    \end{tikzcd}
\end{center}

\begin{lemma}\label{lemma:reduction-functors-composition}
    Reduction-functors are closed under composition.
\end{lemma}
\begin{proof}
    Let
    $ J : \pointedCat \longrightarrow \pointedCatOf{\vocabMap \: \vocabV}$ and
    $ G : \pointedCat \longrightarrow \pointedCatOf{\vocabMap' \: \vocabV}$ be
    reduction-functors. We want to show that $G \circ J$ is also a reduction-functor.
    Using~\cref{obs:reduction-functors-vocab-maps-invariant}, we can lift $G$ to
    $ G : \pointedCatOf{\vocabMap \: \vocabV} \longrightarrow \pointedCatOf{\vocabMap' \: (\vocabMap \: \vocabV)} $.
    % It is a known fact that fully faithful functors are closed under composition.
    Let
    $\reduction{f}$, $\reduction{g}$ be the game reductions for $J$, $G$,
    respectively. Then the action on objects for $G \circ J$ is defined as follows:

    \begin{align*}
        &G \comp J : \pointedCat \longrightarrow \pointedCatOf{\vocabMap' \: (\vocabMap \: \vocabV)} \\
        (&G \comp J) \; \structPair{I}{d} \longmapsto
     \structPair{(\reductionI{g} \circ \reductionI{f}) \: I}{\, (\reductionR{g} \circ \reductionR{f}) \: d}.
    \end{align*}

    \begin{center}
    \begin{tikzcd}[row sep=huge, column sep=huge]
        {\pointedCat} \arrow[d, "J"'] \arrow[rd, "G \; \comp \; J"] &    \\
        {\pointedCatOf{\vocabMap \: \vocabV}} \arrow[r, "G_\vocabMap"']            & {\pointedCatOf{\vocabMap' \: (\vocabMap \: \vocabV)}}
    \end{tikzcd}
    \end{center}

    From~\cref{prop:vocab-maps-inv-composition}, we get that the obtained composition
    is still invariant over vocabulary maps.
\end{proof}

\subsection{Comonadic semantics for extensions}

Having defined appropriate notions and tools, we now present the way to obtain
game semantics for an arbitrary sublogic $\ALC \subseteq \DL{L}\Phi \subseteq \ALCOIbSelf$
by the use of relative comonads.

Let $J_\Phi \triangleq \bigcomp_{\theta \in \Phi} J_\theta$ be a family of
functors indexed by $\Phi$ where $J_\genericExt$ are
$(\reduction{f}_\genericExt, \vocabMap_\genericExt)$-reduction-functors and the
operator $\bigcomp$ iterates over the extensions and composes the functors
together in $\extTuple$ order. It follows
from~\cref{lemma:reduction-functors-composition} that for a fixed $\Phi$, the
functor $J_\Phi : \pointedCatOf{\vocabV} \longrightarrow \pointedCatOf{\vocabV^\Phi}$ is also a reduction-functor.

\begin{proposition}[$\ALC\Phi$-comonad]
    The game comonad $\comonadFunctorOver{\Phi}$ is a $(\comonadFunctor \circ J_\Phi)$-relative-comonad.
\end{proposition}
\begin{proof}
    We know that $J_\Phi : \pointedCatOf{\vocabV} \longrightarrow \pointedCatOf{\vocabV^\Phi}$
    is a functor. From~\cref{prop:DL-is-a-comonad}, we know that
    $\comonadFunctorOver{\Phi} : \pointedCatOf{\vocabV} \longrightarrow \pointedCatOf{\vocabV}$
    is a comonad on $\pointedCatOf{\vocabV}$. Applying~\cref{obs:comonad-vocab-maps-invariant},
    we get ${\comonadFunctorOver{\Phi}}_{(-)^\Phi} : \pointedCatOf{\vocabV^\Phi} \longrightarrow \pointedCatOf{\vocabV^\Phi}$
    which is a comonad on the codomain of $J_\Phi$. Hence, by definition,
    $\comonadFunctorOver{\Phi}$ is a relative comonad.
\end{proof}

With that, we arrive at the concluding lemma which shall guide us to the final result.

\begin{lemma}\label{lemma:phi-bnf-game-is-phi-bisim}
    Let $k \in \Nomega$ and let $\Phi \subseteq \extSet$.
    Given pointed interpretations $\structPairA$ and $\structPairB$, the
    $\backAndForthGameOf{\Phi}$ game for the $\comonadFunctorOver{\Phi}$
    relative comonad is equivalent to the $k$-round $\ALC\Phi(\vocabV)$-bisimulation
    game played on $\structPairA$ and $\structPairB$.
\end{lemma}
\begin{proof}
    By~\cref{theorem:sublogics-game-reduction}, it suffices to show that $\backAndForthGameOf{\Phi}$ is
    equivalent to $\ALC(\vocabV^\Phi)$-bisimulation game between
    $(\reductionI{f}_\Phi \: \interI, \reductionR{f}_\Phi \: d)$ and
    $(\reductionI{f}_\Phi \: \interJ, \reductionR{f}_\Phi \: e)$.
    Recall that the positions in the $\backAndForthGameOf{\Phi}$ are pairs
    $(s, t) \in \comonadStructProdOver{\Phi}$.
    By unfolding the definition of $\comonadFunctorOver{\Phi}$, we get that
    it corresponds to a product of unravelings
    $(\reduction{f}_\Phi \: \interI, d) \times (\reduction{f}_\Phi \: \interJ, e)$.
    Hence, $s$ and $t$ are sequences of the form $\sequence$, where
    $\alpha_i \in \sigma^\Phi_r$ and $a_i \in \DeltaI \lor a_i \in \DeltaJ$ for
    $1 \leq i \leq j$.
    An attentive reader can already notice that it is the same
    as positions in the $\ALC(\vocabV)$-bisimulation game by definition in~\cref{section:bisim-games}.
    What remains to be shown is that the winning conditions coincide.
    Note that after applying~\cref{theorem:sublogics-game-reduction} we
    are playing the $\ALC$-bisimulation game, and thus the same inductive reasoning
    applies as in~\cref{bnf-game-is-bisimiulation-theorem} which concludes the proof.

\end{proof}

For the readers that are still alive and managed to get to this point,
we have finally arrived at the heart of our result. This is summarised by the
following theorem, which is an immediate corollary
from~\cref{fact:games-bisimulations-and-concept-eq},~\cref{lemma:phi-bnf-game-is-phi-bisim}
and~\cref{bnf-game-theorem}. %The last item is the new one.

\begin{theorem}
        For any $k \in \mathbb{N} \cup \{ \omega \}$ and a logic $\DL{L}\Phi$ between $\ALC$ and $\ALCOIbSelf$,   t.f.a.e.:
    \begin{itemize}\itemsep0em
        \item Duplicator has the winning strategy in the $k$-round $\DL{L}\Phi(\vocabV)$-bisimulation-game
             on~$(\interI, \elD ; \interJ, \elE)$,
        \item There is an $\DL{L}\Phi(\vocabV)$-$k$-bisimulation
            $\mathcal{Z}$ between $\interI$ and $\interJ$ such that $\mathcal{Z}(\elD, \elE)$,
        \item $(\interI, \elD) \equiv_k^{\DL{L}\Phi(\vocabV)} (\interJ, \elE)$,
        \item $\structPairA \leftrightarrow_k^{\comonadFunctorBare^\Phi} \structPairB$.
    \end{itemize}
\end{theorem}

%% file: 7-conclusions.tex
\section{Conclusions}\label{section:conclusions}

This paper provides yet another view on bisimulation games used in the
description logic setting, via the lenses of comonadic semantics, as well as
another nail for the comonads hammer developed in recent years.

We have tweaked modal comonad~\cite{relating-struct-and-pow} to match description
logic's setting of interpretations, and devised a composable and extensible
way of tackling logic extensions via
% fully faithful
reduction functors and
relative monads~\cite{monads-need-not-be-endofunctors}.
We now shall discuss the potential directions of what can be done next.

\subsection{Incorporating other known DL extensions}

There wo more $\ALC$ extensions that caught our attention, namely,
counting capabilities and universal role. Following the way graded modalities
were handled in~\cite{relating-struct-and-pow}, we believe that $\ALCQ$,
an extension with counting capabilities, can be encoded by taking
isomorphism in the Kleisli category of $\comonadFunctorPhi$ comonad
in place of $\backAndForthEq$ back-and-forth relation. Concerning universal
role, it appears to be expressible by defining a reduction $\trans{U}$
that adds a fresh role $\roler_{U}$ that forms a clique. However, neither
of the ideas has been carefully verified and thus that is yet to be explored.

\subsection{Combinatorial properties}

Another research direction is to investigate combinatorial properties naturally
arising from the coalgebras of the resulting comonad, such as tree width for the
pebbling comonad~\cite{pebbling-comonad} or tree depth for the modal
comonad~\cite{relating-struct-and-pow}. A topic closely related that generalizes
over parameters is the examination of $\comonadFunctorPhi$ functor's Kan extension
that should yield discrete density comonad~\cite{discrete-density}.

\subsection{Transcribing known theorems to category theory}

This lies at the core of the meaning and purpose of defining comonadic semantics
for model comparison games. S. Abramsky et. al has given a generalization of
the framework by Arboreal categories and covers~\cite{arboreal-categories},
and we have observed a variety of results arising from a categorical framework
such as new Lov{\'{a}}sz-Type Theorems~\cite{lovasz-and-comonads} or
axiomatic account of Feferman-Vaught-Mostowski theorems~\cite{courcelle-mostowski-comonads}.
A systematic overview of the current state of the art in applying tools from category
theory in finite model theory and descriptive complexity is given in~\cite{emerging-landscape}.
Hence, the most natural direction for the next research project would be to explore
how the description logic comonad could help to generalize or simplify known theorems.